\documentclass[a4paper,12pt]{article}
\pdfoutput=1 

\usepackage{geometry}
\usepackage[margin=10mm]{caption}

\usepackage[utf8]{inputenc}
\usepackage[english]{babel}
\usepackage{amssymb, amsmath}
\usepackage{amsthm}
\usepackage{amsfonts}
\usepackage{mathrsfs}
\usepackage{setspace}
\usepackage{graphicx}
\usepackage{subcaption}
\usepackage{slashed}
\usepackage{array,booktabs}
\usepackage{lmodern}
\usepackage{tikz}
\usepackage{microtype}
\usepackage{hyperref}
\usepackage{authblk}
\usepackage{appendix}

\newtheorem{theorem}{Theorem}
\newtheorem{lemma}{Lemma}[section]

\theoremstyle{definition}
\newtheorem{remark}{Remark}[section]

\makeatletter
\renewenvironment{proof}[1][\proofname]{%
   \par\pushQED{\qed}\normalfont%
   \topsep6\p@\@plus6\p@\relax
   \trivlist\item[\hskip\labelsep\bfseries#1\@addpunct{.}]%
   \ignorespaces
}{%
   \popQED\endtrivlist\@endpefalse
}

\newcommand{\p}{\partial}
\newcommand{\ups}{\upsilon}
\newcommand{\beq}{\begin{equation}}
\newcommand{\eeq}{\end{equation}}
\newcommand{\eeeem}{\end{multline}}
\newcommand{\bem}{\begin{multline}}
\newcommand{\bqa} {\begin{eqnarray}}
\newcommand{\eqa} {\end{eqnarray}}
\newcommand{\eps}{\varepsilon}

\newcommand{\bmul}{\begin{multline}}
\newcommand{\emul}{\end{multline}}
\DeclareMathOperator{\End}{End}

\DeclareMathOperator{\CT}{CT}

\DeclareMathOperator{\ad}{ad}
\DeclareMathOperator{\dist}{dist}

\def \ra {\rightarrow}

\newcommand{\CC}{{\mathcal C}}
\newcommand{\CA}{{\mathcal A}}
\newcommand{\CB}{{\mathcal B}}
\renewcommand{\CT}{{\mathcal T}}
\newcommand{\CH}{{\mathcal H}}

\newcommand{\CI}{{\mathcal I}}
\newcommand{\CO}{{\mathcal O}}
\newcommand{\CU}{{\mathcal U}}
\newcommand{\CV}{{\mathcal V}}
\newcommand{\CS}{{\mathcal S}}
\newcommand{\CW}{{\mathcal W}}

\newcommand{\ZZ}{{\mathbb Z}}
\newcommand{\RR}{{\mathbb R}}

\newcommand{\SA}{{\mathscr A}}
\newcommand{\SI}{{\mathscr I}}
\newcommand{\SAl}{{{\mathscr A}_\ell}}
\newcommand{\SAal}{{{\mathscr A}_{a\ell}}}

\newcommand{\OL}{{O(L^{-\infty})}}
\newcommand{\Or}{{O(r^{-\infty})}}

\newcommand{\Ot}{{O(|t|^{-\infty})}}
\newcommand{\HP}{{G_0}}

\newcommand{\Ups}{\Upsilon}
\newcommand{\tUps}{{\tilde\Upsilon}}
\newcommand{\tPsi}{{\tilde\Psi}}
\newcommand{\tPi}{{\tilde\Pi}}
\newcommand{\upsp}{{\upsilon^{(+)}}}

\def \l {\left(}
\def \r {\right)}

\def \lal {\langle}
\def \ral {\rangle}

\title{Hall conductance and the statistics of flux insertions in gapped interacting lattice systems}

\author{Anton Kapustin, Nikita Sopenko \smallskip\\ 
{\it California Institute of Technology, Pasadena, CA 91125, USA}}

\begin{document}

\maketitle

\abstract{We study charge transport for zero-temperature infinite-volume gapped lattice systems in two dimensions with short-range interactions. We show that the Hall conductance is locally computable and is the same for all systems which are in the same gapped phase. We provide a rigorous versions of Laughlin's flux-insertion argument which shows that for short-range entangled systems the Hall conductance is an integer multiple of $e^2/h$. We show that the Hall conductance determines the statistics of flux insertions. For bosonic short-range entangled systems, this implies that the Hall conductance is an even multiple of $e^2/h$. Finally, we adapt a proof of quantization of the Thouless charge pump to the case of infinite-volume gapped lattice systems in one dimension.}

\section{Introduction}

Quantization of the Hall conductance of insulating two-dimensional materials at low temperatures is one of the most remarkable phenomena in condensed matter physics. Starting with the seminal work of Laughlin \cite{Laughlin}, many theoretical explanations of this phenomenon have been proposed which vary in their assumptions and degree of rigor. For systems of non-interacting  fermions with either an energy gap or a mobility gap there are several proofs that  zero-temperature Hall conductance $\sigma_{Hall}$ times $2\pi\hbar/e^2$ is an integer in the infinite-volume limit \cite{TKNN,bellissard1985,bellissard1986,bellissard1994,avron1994}. In particular, Laughlin's flux-insertion argument was made rigorous in this setting by Avron, Seiler and Simon \cite{avron1994}. The case of interacting systems is more involved, since quantization of Hall conductance generally holds only in the absence of topological order. Much progress can be made using the relation between $\sigma_{Hall}$ and the curvature of the Berry connection of the system compactified on a large torus. This line of work originated with Avron and Seiler \cite{avronseiler} and culminated in the proof by Hastings and Michalakis that for a gapped system on a large torus with a non-degenerate ground state the difference between $2\pi\hbar\sigma_{Hall}/e^2$ and the nearest integer is almost-exponentially small in the size of the system $L$, as defined in \cite{hastingsmichalakis} (see also \cite{bachmannetal}).

% \footnote{More precisely the difference between $2\pi\hbar\sigma_{Hall}/e^2$ and the nearest integer is almost-exponentially small in the size of the system $L$, as defined in \cite{hastingsmichalakis}.}.

Despite all the progress in understanding quantization of the Hall conductance in the interacting case, there is still room for improvement. From a modern perspective, $\sigma_{Hall}$ is a topological invariant of a quantum phase of gapped 2d systems with a $U(1)$ symmetry. Since the distinction between phases becomes sharp only in infinite volume, it would be desirable to have a formalism which can deal with systems on a 2d Euclidean space rather than a torus. Within such a formalism it should be possible to see that in a gapped system $\sigma_{Hall}$ is locally computable, i.e. can be approximated well by an expectation value of a local observable. This would clarify the role of $\sigma_{Hall}$ as an obstruction to having a gapped edge. It should also be possible to prove that $\sigma_{Hall}$ is the same for all systems in the same gapped phase.

Another recent development inspired by quantum information theory is the viewpoint that different quantum phases of matter are  distinguished by different patterns of entanglement in the ground state \cite{QImeets}. If this is the case, then it should be possible to extract the value of $\sigma_{Hall}$ from the infinite-volume ground state, without specifying a concrete Hamiltonian. Indeed, in the case of non-interacting fermions, it is well-known how to extract the zero-temperature Hall conductance using only the projector to the ground state, see e.g. \cite{avron1994,kitaev2006anyons}. From this viewpoint, a system without topological order is a system with only short-range entanglement. One convenient definition of this notion (recalled in Section \ref{sec:InvPhases}) was proposed by A. Kitaev \cite{kitaevInv} under the name {\it an invertible gapped  phase}. Gapped systems of free fermions are invertible \cite{hastings:free}, but there are many interacting systems of this kind as well. It should be possible to prove quantization of $2\pi\hbar\sigma_{Hall}/e^2$ for an arbitrary system in an invertible gapped phase.

Finally, it has been argued using the statistics of flux insertions that for bosonic short-range entangled gapped 2d systems with a $U(1)$ symmetry $2\pi \hbar\sigma_{Hall}/e^2$ is always an {\it even} integer \cite{LevinSenthil}. It would be desirable to prove this, and studying flux-insertion for systems on a 2d Euclidean space is a natural approach. 

Another phenomenon closely related to the quantization of the Hall conductance is quantized charge transport in gapped non-degenerate 1d systems with a $U(1)$ symmetry. It was argued by Thouless \cite{Thouless} that one can attach a numerical invariant to a loop in the space such systems. This invariant is integral in the infinite-volume limit and is equal to the net charge pumped through a section of the system as it adiabatically cycles through the loop. Some of the questions mentioned above (such as local computability and independence of the particular Hamiltonian) can be also asked about the Thouless pump invariant. Recently Bachmann et al. \cite{bachmann2019many} proved approximate quantization of the Thouless pump invariant for gapped 1d systems on a large circle using quasi-adiabatic evolution and sub-exponential filter functions \cite{hastings2010quasi}. This proof makes explicit that the Thouless pump invariant is locally computable. One can hope to use similar methods to achieve a better understanding of the Hall conductance of gapped 2d systems.

In this paper we use the methods of \cite{bachmann2019many,hastings2010quasi,bachmann2020many,bachmann2020rational} to study both the Thouless pump invariant and the Hall conductance for gapped lattice systems in infinite volume. In the case of 1d systems, we merely adapt the approach of \cite{bachmann2019many} to the infinite-volume setting. In the case of 2d systems our main results can be summarized as follows (using the units where $\hbar=e^2=1$):
\begin{theorem}
Zero-temperature Hall conductance is locally computable for any 2d lattice system (either bosonic or fermionic) with an on-site $U(1)$ symmetry, exponentially decaying interactions, and a unique  gapped ground state on $\RR^2$. It is the same for all systems in the same gapped phase. If the system is an invertible gapped phase, $2\pi \sigma_{Hall}\in\ZZ$. If the system is also bosonic, then $2\pi \sigma_{Hall}\in 2\ZZ$.
\end{theorem}
% \footnote{A phase is an equivalence class of gapped systems with respect to deformations of the Hamiltonian, see Section \ref{sec:InvPhases}  for more details.}
In fact, we give two proofs of the quantization of Hall conductance. One of them is a version of Laughlin's flux-insertion argument adapted to the case of systems in an invertible phase. The other one uses the relation between Hall conductance and the statistics of flux insertions  \cite{LevinSenthil}. We make this relation precise by defining local operators which transport flux and showing that their large-scale properties are controlled by $\sigma_{Hall}$. We also show that Hall conductance can be determined from the ground-state of a gapped Hamiltonian without specifying the Hamiltonian and that it is invariant under a certain equivalence relation on states. This equivalence relation is induced by automorphisms of the algebra of observables which are "fuzzy" analogues of finite-depth local unitary quantum circuits.

The paper is organized as follows. In Section 2 we give  definitions and describe some constructions pertaining to gapped lattice systems with $U(1)$ symmetry in any dimension. These definitions and constructions are used throughout the remainder of the paper. In particular, we define the notions of a gapped  phase and an invertible gapped phase. In Section 3 we adapt the results of  \cite{bachmann2019many,bachmann2020rational} on charge pumping to  infinite-volume systems. In Section 4 we study Hall conductance. We show that zero-temperature Hall conductance is locally computable, define vortex states and transport of vortices, show that the Hall conductance controls the statistics of vortices, and use this to argue that for systems in an invertible phase $2\pi\sigma_{Hall}$ is an integer, while for bosonic systems in an invertible phase $2\pi\sigma_{Hall}$ is an even integer. We also show that $\sigma_{Hall}$ is determined by the state and is invariant under a certain class of automorphisms with good locality properties (locally generated automorphisms). In Appendix A we collect some technical results used in the paper. In Appendix B we present a version of  Lauglin's  flux-insertion argument which shows that $2\pi\sigma_{Hall}$ is quantized for interacting systems in an invertible phase. It uses some of the results of Sections 3 and 4. In Appendix C we prove triviality of superselection sectors of certain states produced from a factorized state by locally generated automorphisms.\\

% \noindent
% {\bf Acknowledgements:}
% This research was supported in part by the U.S.\ Department of Energy, Office of Science, Office of High Energy Physics, under Award Number DE-SC0011632. A.K. was also supported by the Simons Investigator Award. N.S. gratefully acknowledges the support of the Dominic Orr Fellowship at Caltech. \\

% \noindent
% {\bf Data availability statement:}
% Data sharing is not applicable to this article as no new data were created or analyzed in this study.

\section{Preliminaries}

\subsection{Basic definitions}

In this paper we study lattice many-body systems (bosonic or fermionic) in $d$ dimensions. We follow the operator-algebraic framework described in the monograph \cite{bratteli2012operator}.  A lattice in $d$ dimensions is an infinite subset $\Lambda$ of the Euclidean space $\RR^d$ which is uniformly discrete (that is, there is an $r>0$ such that for any two distinct $j,k\in\Lambda$ $\dist(j,k)\geq r$) and uniformly filling (that is, there is an $r'>0$ such that for any $x\in\RR^d$ $\dist(x,\Lambda)\leq r'$). The algebra of observables on a site $j\in\Lambda$ is a matrix algebra $\SA_j=\End(\CH_j)$, where $\CH_j$ is a finite-dimensional complex vector space with the dimension $\dim \CH_j = d_j^2$ that grows at most polynomially with the distance from the origin. In the fermionic case the vector spaces $\CH_j$ are $\ZZ_2$-graded by fermion parity, so the algebras  $\CA_j$ are also $\ZZ_2$-graded. In the bosonic case for any finite subset $\Gamma\subset\Lambda$ we let $\SA_\Gamma=\otimes_{j\in\Gamma} \SA_j$. Then for any inclusion of finite subsets  $\Gamma\subset \Gamma'$ we have an obvious injective homomorphism $\SA_\Gamma\ra\SA_{\Gamma'}$, and the algebras $\SA_\Gamma$ form a directed system over the directed set of finite subsets of $\Lambda$.  A normed  $*$-algebra $\SAl$ of local observables is defined as the direct limit of this directed system:
\beq\label{dirlim}
%\SAl=\underset{\Gamma}{\varinjlim}\, \SA_\Gamma
\SAl= \cup_{\Gamma} \SA_\Gamma .
\eeq 
Then the $C^*$-algebra $\SA$ is defined as the norm completion of $\SAl$. Elements of $\SA$ will be referred to as quasi-local observables, or simply observables. In the fermionic case, we do the same, except we define $\SA_\Gamma$ using the graded tensor product. Also, since all observables appearing in this paper will be bosonic, in the fermionic case by a (quasi-local) observable we will always mean an even element of $\SA$.

A quasi-local observable $\CA$ is called local with a compact localization set $\Gamma\subset\Lambda$ if $\CA\in\SA_\Gamma.$ We may also say that such an $\CA$ is localized on $\Gamma$. Equivalently, $\CA$ commutes with any observable $\CB\in \SA_j$ for any site $j \in \bar{\Gamma} = \Lambda \backslash \Gamma$. We call an observable $\CA\in\SA$ almost local if it can be approximated well by a local observable. To be more precise, let us denote by $B_r(j)$ a ball of radius $r>0$ with the center at $j\in\Lambda$. That is, $B_r(j)=\{k\in\Lambda, \dist(k,j)<r\}.$ Also, let us pick a monotonically decreasing positive (MDP) function $a(r)$ on $\RR_+=[0,+\infty)$ which has superpolynomial decay, i.e. it is of order $\Or$. An observable $\CA$ will be called $a$-localized on a site $j$ if for any $r>0$ there is a local observable $\CA^{(r)}\in\SA_{B_r(j)}$ such that $||\CA-\CA^{(r)}||\leq||\CA|| a(r)$. An observable will be called almost local if it is $a$-localized on $j$ for some MDP function $a(r)=\Or$ and some $j\in\Lambda$. Almost local observables approximately commute when they are localized far from each other. More precisely, if $\CA$ is $a$-localized on $j$ and $\CB$ is $b$-localized on $k$, then $||[\CA,\CB]||\leq ||\CA||\cdot ||\CB|| c(\dist(j,k)),$ where $c(r)=2(a(r/2)+b(r/2)+3 a(r/2) b(r/2))=\Or.$ Almost local observables can also be characterized in terms of their commutators with local observables, see Lemma \ref{lma:approximation}. 
We will denote the $*$-algebra of almost local observables $\SAal$. By definition, $\SAl$ is a dense sub-algebra of $\SA$. Since $\SAal\supset \SAl$, $\SAal$ is also dense in $\SA$.

A Hamiltonian $H$ for $\SA$ is a formal sum
\beq
H = \sum_{j \in \Lambda} H_{j},
\eeq
where the interactions $H_j\in\SA$ are assumed to be self-adjoint, uniformly bounded and exponentially decaying, i.e. there are $J>0$ and  $R>0$ such that $\|H_{j}\| \leq J$ and for any local $\CA$ localized on any site $k$ we have $\|[H_{j},\CA]\|\leq J \|\CA\| e^{-\text{dist}(j,k)/R}.$ It is well-known \cite{bratteli2012operator2} that such an $H$ gives rise to a strongly-continuous one-parameter family of automorphisms of $\SA$ (the time evolution). We denote this family $\tau_t$, $t\in\RR$. 

More generally, we will consider formal linear combination of the form  $F=\sum_{j\in\Lambda} F_{j}$, where $F_j\in\SA$ satisfy the following two  conditions. First, the observables $F_j$ are uniformly bounded, i.e. there exists $C>0$ such that $\|F_j\|\leq C$ for all $j$. Second, there is an MDP function $f(r)$ on $\RR_+$ such that $f(r)=\Or$ and for all $j\in\Lambda$ the observable $F_j$ is $f$-localized on $j$. In particular, $F_j$ is an almost local observable.  Such a formal linear combination $F$ will be called a 0-chain. If we want to specify the function $f$, we will say that $F$ is an $f$-local 0-chain. The Hamiltonian $H$ is an example of a 0-chain.

We will be also using 1-chains (currents) $J_{jk}$ and 2-chains (2-currents) $M_{jkl}$. A current is a skew-symmetric function  $J:\Lambda\times\Lambda\ra \SA$ satisfying three conditions. First, $\|J_{jk}\|$ are uniformly bounded: $\|J_{jk}\| \leq C$ for some $C>0$. Second, there exists an MDP function $f(r)=\Or$ such that for any $j,k\in\Lambda$  $\|J_{jk}\| \leq C \, f(\dist(j,k)).$ Third, for any $j,k\in\Lambda$ and any $r>0$ there is a local $J^{(r)}_{jk}\in\SA_{B_r(j)\cup B_r(k)}$
such that $||J_{jk}-J^{(r)}_{jk}||\leq C f(r).$ If we want to specify the function $f$, we will say that $J$ is an $f$-local current. An $n$-chain (or $n$-current) for general $n$ is defined similarly: it is a skew-symmetric function $M_{j_0\ldots j_n}$ on the Cartesian product of $n+1$ copies of $\Lambda$ which is valued in $\SA$, is uniformly bounded, is rapidly decaying when the arguments are far apart, and such that for any $r>0$  $M_{j_0\ldots j_n}$ can be approximated well by a local observable on $B_r(j_0)\cup \ldots \cup B_r(j_n)$. Such objects were introduced by A. Kitaev \cite{kitaev2006anyons} and have found several interesting applications \cite{kapustin2019thermal,kapustin2020higherA,kapustin2020higherB}. In this paper we only use $n$-chains with $n\leq 2$.

We define a linear map $\partial$ from the space of currents to the space of  0-chains:
\beq
(\p J)_j = \sum_{k \in \Lambda} J_{jk},
\eeq
and more generally from the space of $n$-chains to the space of $(n-1)$-chains:
\beq
(\p M)_{j_0\ldots j_{n-1}} = \sum_{j_n \in \Lambda} M_{j_0\ldots j_n}.
\eeq
The definition of $n$-chains was chosen so that these maps are well-defined. Also, we have $\p^2=0$.

To any 0-chain $F$ and any compact set $\Gamma$ we attach an almost local  observable $F_\Gamma=\sum_{j\in\Gamma} F_j$. Similarly, we denote
\beq\label{defJAB}
J_{AB} = \sum_{j \in A} \sum_{k \in B} J_{jk},
\eeq
\beq
M_{ABC} = \sum_{j \in A} \sum_{k \in B}  \sum_{l \in C} M_{jkl} .
\eeq
Here $A,B,C$ are some subsets of $\Lambda$ for which the above sums are  convergent. It is important that these expressions can be well-defined observables even if the subsets $A,B,C$ are not compact. For example, if $d=1$ and $A$ and $B$ are complementary half-lines in $\RR$, the sum defining $J_{AB}$ is convergent and defines an almost local observable which is localized near the boundary of $A$ and $B$. 

For a general 0-chain $F$, the formal expression $F=\sum_{j\in\Lambda} F_j$ is not a well-defined element of $\SA$. But it gives rise to a well-defined derivation $\ad_F=[F,\cdot]$ of the algebra  $\SAal$ defined by
\begin{equation}
\ad_F(\CA)=\sum_{j\in\Gamma}[F_j,\CA].
\end{equation}
This derivation is unbounded, therefore  cannot be extended to the whole $\SA$. Its main use is to define a class of automorphisms of the algebra $\SA$ with good locality properties. Let $F(s)=\sum_{j\in\Lambda} F_j(s)$ be a self-adjoint 0-chain which is a differentiable function of a parameter $s\in [0,1]$. 
Then one can define a family of automorphisms $\alpha_F(s)$ of $\SA$ by solving the equations 
\beq
-i \frac{d}{ds} \alpha_{F}(s)(\CA) = \alpha_F(s)\left([F(s),\CA]\right), \,\,\,\, \alpha_{F}(0)=1,
\eeq
for all $\CA\in\SAal$ and then extending to the whole $\SA$ by continuity.
These equation can be shown to have a unique solution \cite{bratteli2012operator2}. Using the Lieb-Robinson bound, one can show that for any self-adjoint 0-chain $F$ and any $s$ the automorphism $\alpha_F(s)$ maps $\SAal$ to $\SAal$. More precisely, if  $F$ is a self-adjoint $f$-local 0-chain, and $\CA\in\SA$ is $a$-localized on a site $j$ for some MDP function $a(r)=\Or$, then $\alpha_F(s)(\CA)$ is $h$-localized on $j$. Here $h(r)=\Or$ is an MDP function which depends on $\CA$ only through $a$.  For a proof, see Appendix \ref{app:lemmas}. Note also that $\alpha_{F}(ts)=\alpha_{t F}(s)$ and $\alpha_{-F}(s)=\alpha_F(s)^{-1}$. If $F$ is independent of $s$, then the automorphisms $\alpha_F(s)$ satisfy $\alpha_F(t)\alpha_F(s)=\alpha_F(t+s)$.
We will call an automorphism $\alpha$ of $\SA$ locally generated if there exists an $s$-dependent self-adjoint 0-chain $F$ such that $\alpha=\alpha_F(1).$ In what follows, by a 0-chain we will always mean a self-adjoint 0-chain.

Above we have defined localization sets for local 
observables. One can also define approximate localization sets for 0-chains. 
%An $L$-thickening of $\Gamma$ is a subset of $\Lambda$ consisting of all points which are within distance $L$ from $\Gamma$. We will denote it $\CT_L\Gamma$.
A 0-chain $F$ is said to be approximately localized on $\Gamma\subset\Lambda$ if there is a $C>0$ such that for any $j\in\Lambda$, any MDP function $a$ such that $a(r)=\Or$, and any $\CA$ which $a$-localized on $j$ one has $\|[F,\CA]\|\leq C \|\CA\| h({\rm dist}(j,\Gamma))$, where the MDP function $h(r)=\Or$ depends on $\CA$ only through $a$. Given an arbitrary 0-chain $F=\sum_j F_j$, one  can construct a truncated 0-chain $F_\Gamma$ approximately   localized on any $\Gamma$ by letting $F_\Gamma=\sum_{j\in\Gamma} F_j$. If $\Gamma$ is finite, $F_\Gamma$ is an almost local observable, otherwise it is a 0-chain approximately localized on $\Gamma$. Similarly, given a current $J_{jk}$ and any two subsets $A,B,$ we can interpret $J_{AB}$ as a 0-chain with components $J_{AB,j}=\sum_{k\in B} J_{jk}$ for $j\in A $ and $J_{AB,j}=0$ for $j\in\bar A$  even if the sum (\ref{defJAB}) is not convergent. We will also say that an automorphism $\alpha$ is approximately localized on $\Gamma\subset\Lambda$ if for any $j\in\Lambda$ and any $\CA$ which is $a$-localized on $j$ one  has $\|\alpha(\CA)-\CA\|\leq \|\CA\| h({\rm dist}(j,\Gamma))$ for some MDP function $h(r)=\Or$ which depends on $\CA$ only through $a$. Lemma \ref{lma:FGA} shows that the action of the automorphism $\alpha_F(s)$ on an observable localized near $j\in\Lambda$ depends only on the behavior of $F$ in the neighborhood of $j$.

We will say that a state $\psi$ on $\SA$ is superpolynomially clustering if there is an MDP function $h(r)=\Or$ such that for any two finite subsets $\Gamma$ and $\Gamma'$ and any two observables $\CA\in\SA_{\Gamma}$ and $\CA'\in\SA_\Gamma'$ one has
\beq
\left| \langle \CA\CA'\rangle_\psi-\langle \CA\rangle_\psi \langle \CA'\rangle_\psi \right| \leq C |\Gamma|\cdot |\Gamma'|\cdot ||\CA||\cdot ||\CA'||\,  h(\dist(\Gamma,\Gamma')).
\eeq
Here $|\Gamma|$ is the number of sites in $\Gamma$.
If $h(r)$ can be chosen to have the form $h(r)=C e^{-r/\xi}$ for some $C>0$ and $\xi>0$, we will say that the state $\psi$ is exponentially clustering.

If $\psi$ is superpolynomially clustering, $\CA$ is $a$-localized on $j\in\Lambda$, and $\CB$ is $b$-localized on $k\in\Lambda$, then $\left| \langle \CA\CB\rangle_\psi-\langle \CA\rangle_\psi \langle \CB\rangle_\psi \right| \leq \|\CA\|\cdot \|\CB\| f(\dist(j,k))$, where the MDP function $f(r)=\Or$ depends on $\CA,\CB$ only through $a$ and $b$.

For a state $\psi$ and an automorphism $\alpha$ one can define a new state $\alpha(\psi)$ by $\langle\CA\rangle_{\alpha(\psi)}=\lal \alpha(\CA) \ral_{\psi}.$  
It is easy to see that if the state $\psi$ is superpolynomially clustering and the automorphism $\alpha$ is locally-generated, then $\alpha(\psi)$ is also superpolynomially clustering.

\subsection{Quasi-adiabatic evolution}\label{QAevolution}

In this paper we will be studying lattice systems with an energy gap. That is, we assume that we are given a pure state $\psi$ on $\SA$ which is a ground-state of $\tau_t$, i.e. for any $\CA \in \SAl$ we have $\lal \CA^* \ad_H(\CA) \ral_{\psi} \geq 0$. In the fermionic case, we also assume that $\langle \CA\rangle_\psi=0$ for any odd $\CA$. We further assume that it is a unique gapped ground-state, in the sense that in the GNS Hilbert space corresponding to $\psi$ the GNS vacuum vector is a unique vector invariant under the unitary evolution corresponding to $\tau_t$, and that $0$ is an isolated eigenvalue of the Hamiltonian in the GNS representation (the generator of the unitary evolution).  A gapped lattice system is a triple $(\SA,H,\psi)$, where $\SA$ is the algebra of observables, $H$ is a Hamiltonian, and $\psi$ is a state, with the properties described above. These properties imply that the state $\psi$ has the exponential clustering property with some characteristic length scale  $\xi$ \cite{HastingsKoma,NachtergaeleSims}.

In the presence of the energy gap one can define  certain useful linear maps from $\SA$ to $\SA$.
For any $\Delta>0$ we choose a continuous function $W_\Delta(t)$ which is a real, odd, bounded,  superpolynomially decaying for large $|t|$, and $\hat{W}_\Delta(\omega) = - \frac{i}{\omega}$ for $|\omega| > \Delta $. Here $\hat W_\Delta$ is the Fourier-transform of $W_\Delta$, $\hat W_\Delta(\omega)=\int e^{i\omega t} W_\Delta(t) dt$. It was shown in \cite{hastings2010quasi} that such function exists. If $H=\sum_j H_j$ is a gapped Hamiltonian with respect to a ground state $\psi$, we pick a $\Delta$ smaller than the gap and for any observable $\CA$ define 
\begin{equation}
\SI_\Delta(\CA):=\int^\infty_{-\infty} W_\Delta(t) \tau_t(\CA) dt.
\end{equation}
The map $\SI_\Delta:\SA\ra\SA$ is a bounded linear map which commutes with conjugation. It also maps $\SAal$ to $\SAal$. Indeed, suppose $a(r)=\Or$ and $\CA$ is $a$-localized on a site $j$. By Lemma \ref{lma:approximation}, it is sufficient to prove that for any $k\in\Lambda$ and any $\CB\in\SA_k$ the norm of the commutator $[\SI_\Delta(\CA),\CB]$ is bounded from above by a quantity of order $\Or$, where $r=\dist(k,j).$ Since one can approximate $\CA$ with an $a(r/3)$ accuracy by a local observable $\CA^{(r/3)}$ localized on $B_{r/3}(j)$, it is sufficient to prove that $[\SI_\Delta(\CA^{(r/3)}),\CB]$ is of order $\Or$. The Lieb-Robinson bound for exponentially decaying interactions \cite{HastingsKoma} implies that one has an estimate $\|[\tau_t(\CA^{(r/3)}),\CB]\|<C r^d \|\CA^{(r/3)}\|\cdot \|\CB\| \exp((v_0 |t|-2r/3)/R_0)$, where $C, v_0,$ and $R_0$ are positive numbers. Since $W_\Delta(t)=\Ot$, this implies the desired result. Also, since $W_\Delta$ is odd, $\langle \SI_\Delta(\CA)\rangle_\psi=0$ for any $\CA\in\SA$.

Using the functional calculus for unbounded operators one can easily see that for any $\CA,\CB\in\SA$ one has identities
\beq\label{Kubo_static}
\langle \SI_\Delta(\CA) \CB\rangle_\psi=i\left\langle 0| \CA \HP \CB\right|0\rangle,\quad \langle \CB \SI_\Delta(\CA) \rangle_\psi=-i\left\langle0| \CB \HP \CA\right|0\rangle,
\eeq
where $|0\rangle$ is a cyclic vector for the GNS representation, observables are identified with their images in this representation, and $\HP=(1-P)\frac{1}{H}(1-P)$ with $P=|0\rangle \langle 0|$. Therefore for any two observables $\CA$, $\CB$ one has
\bqa\label{Kubo_basic_property}
\langle \SI_\Delta(i[H,\CA])\CB\rangle_\psi=\langle \CA \CB\rangle_\psi-\langle \CA\rangle_\psi\langle \CB\rangle_\psi,
\\
\quad \langle \CB \SI_\Delta(i[H,\CA])\rangle_\psi =\langle \CB \CA\rangle_\psi -\langle \CB\rangle_\psi\langle \CA\rangle_\psi.
\eqa
\begin{remark}\label{Wdelta ambig}
While the l.h.s. of eqs. (\ref{Kubo_static}) involves a map $\SI_\Delta$ which depends both on $\Delta$ and the choice of the function $W_\Delta$, the r.h.s. does not depend on either. This is consistent because for any self-adjoint $\CA\in\SA$ changing $W_\Delta$ changes $\SI_\Delta(\CA)$ only by an observable which annihilates the ground state. One can easily check this property using functional calculus.
\end{remark}

We say that a self-adjoint observable $\CA$ does not excite the state $\psi$, if for any observable $\CB$ such that $\langle \CB \rangle_{\psi} = 0$ we have $\langle \CA \CB \rangle_{\psi} = 0$. For brevity we denote this condition by $\langle \CA ... \rangle_{\psi} = 0$. Note that if $\langle \CA ... \rangle_{\psi} = 0$, then $\lal e^{i \CA} ... \ral_{\psi}=0$ and $\lal e^{i\CA}\ral=1$.

If $\psi$ is the ground state of a gapped Hamiltonian, then for any self-adjoint almost local observable $\CA$ we can define a self-adjoint almost local observable $\tilde \CA$ which does not excite $\psi$ by letting
\beq
\tilde \CA:=\CA-\SI_\Delta(i[H,\CA]).
\eeq
This follows easily from (\ref{Kubo_basic_property}).
Although $\tilde\CA$ depends on the choice of the function $W_\Delta$, by the Remark \ref{Wdelta ambig} varying $W_\Delta$ affects $\tilde\CA$ only by a self-adjoint observable which annihilates the ground state.

Consider a differentiable family of Hamiltonians $H(s)=\sum_j H_j(s)$, $s\in [0,1],$ with a unique gapped ground state $\psi_s$ for each $s$.  Let $\Delta$ be a positive number less than the lower bound of the gaps of the Hamiltonians $H(s)$. We define a one-form on $[0,1]$ with values in $\SA$ by 
\beq
G_j (s) ds := \SI_\Delta(d H_j(s)) .
\eeq
The formal sum $G(s) = \sum_{j \in \Lambda} G_j(s)$ is a 0-chain.

The automorphism $\alpha_G(s)$ generated by $G(s) = \sum_j G_j(s)$ implements a quasi-adiabatic evolution introduced in \cite{hastings2004lieb,hastings2005quasiadiabatic}. It was shown in \cite{bachmann2012automorphic} that if $H(0)$ is a Hamiltonian with a gapped ground state which is a limit of finite-volume Hamiltonians with gapped ground states, this automorphism gives an automorphic equivalence of ground states of $H(s)$ for all $s$. The expectation value of an almost local observable $\CA$ in the ground-state of $H(s)$ therefore satisfies
\beq
\frac{d}{ds} \langle \CA \rangle(s) = \langle i[G(s), \CA] \rangle .
\eeq
Using eq. (\ref{Kubo_static}), one can easily see that this is equivalent to the Kubo formula for static linear response (at zero temperature). 
We can also quasi-adiabatically evolve the systems only on some region $S$ by an automorphism $\alpha_{G_{S}}$ generated by $G_{S}(s) = \sum_{j \in S} G_{j}(s)$. Since the ground state $\psi$ is exponentially clustering, the state $\alpha_{G_S}(s)(\psi)$ is superpolynomially clustering.

\subsection{Gapped systems with a $U(1)$ symmetry}\label{Gapped system with U(1) symmetry}

We say that a pair $(\SA,H)$ has an on-site $U(1)$ symmetry if we are given a self-adjoint $0$-chain $Q=\sum_j Q_j$ where $Q_{j}\in \SA_j$ satisfies $\exp(2\pi i Q_j)=1$ and $[Q,H_{j}] =0$ for all $j\in\Lambda$. The 0-chain $Q=\sum_j Q_j$ will be called the charge. The corresponding family of automorphisms $\alpha_Q(s)$ satisfies $\alpha_Q(2\pi)=1$. A gapped lattice system with a $U(1)$ symmetry is a quadruple  $(\SA,H,\psi,Q)$, where $(\SA,H,\psi)$ is a gapped lattice system and the pair $(\SA,H)$ has an on-site $U(1)$ symmetry with charge $Q$. One does not need to require $\psi$ to be invariant with respect to the automorphisms $\alpha_Q(s)$: this follows from the Goldstone theorem (see below). 

If $(\SA,H)$ has an on-site $U(1)$ symmetry with  charge $Q$, one can define a current
\beq
J_{jk} = i [H_k, Q_j] - i [H_{j},Q_k]
\eeq
which satisfies the conservation law $\left.\frac{d}{dt}\right|_0\tau_t\left(Q_j\right) = (\p J)_{j} =  \sum_{k} J_{jk}$.

Suppose an on-site $U(1)$ symmetry preserves an $s$-dependent 0-chain $F(s)$, in the sense that  $[Q,F_j(s)]=0$ for all $j$ and all $s$. Then we can associate to $F(s)$ a current
\beq
T^{F}_{jk} = \int_{0}^{1} \alpha_{F}(s) (i [F_k(s), Q_j] - i [F_j(s), Q_k]) ds.
\eeq
It satisfies
\beq\label{1st property of T}
\alpha_F(1)(Q_j)-Q_j=\left(\partial T^F\right)_j.
\eeq
Therefore for any finite region $\Gamma$ we have
\beq\label{2nd property of T}
\alpha_F(1)(Q_{\Gamma}) - Q_{\Gamma} = T^{F}_{\Gamma \bar{\Gamma}}.
\eeq
Thus the current $T^{F}$ measures the charge transported by the automorphism $\alpha_F(1)$. Note that the r.h.s. of eq. (\ref{2nd property of T}) can be well-defined even if $\Gamma$ is infinite. Then it can serve as a definition of the charge transported from $\bar\Gamma$ to $\Gamma$.

 We say that a state $\psi$ has no local spontaneous symmetry breaking if there is a $U(1)$-invariant self-adjoint current $K_{j k}$, such that 
\beq
\langle (Q - (\p K))_j ... \rangle_{\psi} = 0.
\eeq
This condition implies the absence of spontaneous symmetry breaking, i.e. $\langle [Q,\CA]\rangle_\psi=0$ for any almost local observable $\CA$. Equivalently, $\langle \alpha_Q(s)(\CA)\rangle_\psi=\langle \CA\rangle_\psi$ for any observable $\CA$.
We use the notation
\beq
\tilde{Q}_j := Q_j - (\p K)_j
\eeq
for a modified local charge that does not excite the ground state. For a finite region $\Gamma$ we introduce a unitary observable
\beq
V_{\Gamma}(\phi) := e^{i \phi \tilde{Q}_{\Gamma}} =  e^{i \phi (Q_{\Gamma} - K_{\Gamma \bar{\Gamma}})}.
\eeq
Since $\tilde Q_\Gamma$ does not excite the ground state, $V_\Gamma(\phi)$ satisfies $\langle V_\Gamma(\phi) ... \rangle_\psi=0$.

For a ground state of a gapped Hamiltonian one can always find such a $K_{jk}$ by letting
\beq
K_{jk} = \SI_\Delta(J_{j k}).
\eeq
Therefore such a state does not break $U(1)$ symmetry spontaneously. This is the usual Goldstone theorem. 
Moreover, if we have a state obtained from the ground state by some locally generated $U(1)$-invariant automorphism $\alpha=\alpha_{F}(1)$, then a suitable $K_{jk}$ also exists:
\beq\label{eq:Knew}
K_{jk} = \alpha^{-1}_F(\SI_\Delta(J_{jk}) + T^{F}_{jk}).
\eeq

There are two kinds of ambiguities for $K_{jk}$. First, one can add to $K_{jk}$ any $(\p N)_{jk}$ for any 2-current $N_{jkl}$. Second, one can add any 1-current $K'_{jk}$ that does not excite the ground state. Varying the function $W_{\Delta}(t)$ and $\Delta$ in $\CI_{\Delta}(J_{jk})$ results in the second kind of ambiguity. All physical quantities defined below do not depend on these ambiguities.   
 
\begin{remark}\label{rmk:Poincare}
If for a ground state of a gapped Hamiltonian it is true that  $\lal (\p M)_{j_1 j_2... j_{n-1}} \, ... \ral = 0$ for some $(n-1)$-chain $M$ for $j_a \in S$ for some region $S$, then there exists an $n$-chain $N_{j_1 j_2 ... j_{n}}$ such that
\beq\label{KN}
\lal (M_{{j_1 j_2 ... j_{n}}} - (\p N)_{j_1 j_2 ... j_{n}}) ... \ral  = 0,\quad \text{if all } j_a \in S.
\eeq
Indeed, one can take
\beq
N_{j_1 ... j_{n+1}} = \SI_\Delta(i[\tilde{H}_{j_1},M_{j_2 ... j_{n+1}}]) + \text{cyclic permutations of }\{j_1,...,j_{n+1}\}
\eeq
where
\beq
\tilde{H}_j = H_j - \SI_\Delta(i[H,H_j]).
\eeq

In particular if $\lal (\p K)_{jk} ... \ral = 0$, then there is a 2-current $N_{jkl}$, such that $\lal (K_{jk}-(\p N)_{jk}) ... \ral=0$. This implies that the only ambiguities in the current $K$ are those noted above.

\end{remark}

\subsection{Invertible phases}\label{sec:InvPhases}

A gapped lattice system $(\SA,H,\psi)$ is said to be trivial if for all $j\in\Lambda$ $H_j$ is a local operator localized on $j$, and $\psi$ is factorized, i.e. $\langle \CA\CB\rangle_\psi=\langle\CA\rangle_\psi\langle\CB\rangle_\psi$ whenever $\CA$ and $\CB$ are local observables localized on two different sites $j,k\in\Lambda$.
Two gapped lattice systems with the same algebra of observables are said to be in the same phase if there is a differentiable path of gapped lattice systems connecting them. A bosonic gapped lattice system $(\SA,H,\psi)$ is said to be in an invertible phase, if there is another bosonic gapped lattice system $(\SA',H',\psi')$ (``the inverse system'') such that a combined system is in the trivial phase. That is, there is a path of gapped lattice systems  $(\SA\otimes\SA', H(s),\Psi(s))$ such that $H(s)$ is differentiable, $H(0)=H\otimes 1+1\otimes H'$, $\psi(0)=\psi\otimes\psi'$, and the system $(\SA\otimes\SA',H(1),\Psi(1))$ is trivial. Note that by the results of \cite{bachmann2012automorphic} the state $\Psi(0)$ of the combined system and the factorized state $\Psi(1)$ are automorphically equivalent. In the fermionic case the definition is the same, except one uses the graded tensor product.

If the system $(\SA,H,\psi)$ has a $U(1)$ symmetry, one may define a more restricted notion of an invertible phase by requiring the inverse system also to have a $U(1)$ symmetry and the path of systems interpolating between the composite system and the trivial system to preserve the diagonal $U(1)$ symmetry. We do not use this more restricted notion of an invertible phase in this paper.

\section{Charge pumping}\label{sec:chargepumping}

\subsection{General considerations} \label{ssec:chargepumping1}

Suppose we have an automorphism $\alpha=\alpha_{F}(1)$ locally generated by some $U(1)$-invariant self-adjoint 0-chain $F(s) = \sum_j F_j(s)$. There is a current $T^{F}_{jk}$ that measures the charge transported by this automorphism. In the following we omit the subscript $F$ and use the notation $T_{jk}$ for this current. Our goal in this section is to show that for certain subsets $A,B\subset\RR^d$ the quantity $\langle T_{AB}\rangle$ is approximately quantized if the automorphism $\alpha$ preserves the ground state.

By a region we will mean an embedded $d$-dimensional submanifold of $\RR^d$ whose boundary has a finite number of connected components and does not intersect $\Lambda$. By a slight abuse of notation, we will identify a region and its intersection with the lattice $\Lambda$. Further, we will consider sequences of regions and observables labeled by some parameter $L$ taking values in positive integers and study their behavior in the limit $L\ra\infty$. A sequence of regions $\Gamma_L$ will be called large if the distance between any two connected components of $\partial\Gamma_L$ is $O(L)$. For simplicity, we will shorten "a large sequence of regions" to "a large region", keeping in mind that all regions depend on a parameter $L$. Similarly, a sequence of almost local observables labeled by $L$ will be identified with an $L$-dependent almost local observable.

All $L$-dependent almost local observables considered here will have the form $\CA_L=F_{\Gamma_L}$ or $\CB_L=J_{A_L B_L}$, or some function of these. Here $F$ is an $L$-independent 0-chain, $J$ is an $L$-independent current, $\Gamma_L$ is a large compact region, and $A_L, B_L$ are large compact regions all of whose boundary components are either disjoint or coinciding. The norm of such observables can be bounded from above by functions of $L$ which are $O(L^d)$ or $O(L^{d-1})$ for large $L$. Thus $\|[\CA_L,\CC]\|=\OL$ for any $L$-independent $\CC\in\SA_j$ with $\dist(j,\Gamma_L)=O(L)$ and 
$\|[\CB_L,\CC]\|=\OL$ for any $L$-independent $\CC\in\SA_j$ with $\dist(j,A_L\bigcap B_L)=O(L)$. Also, for any two $L$-dependent regions $\Gamma_L$ and $\Gamma'_L$ such that the distance between them is $O(L)$ and any two $L$-independent 0-chains $F,F'$ the corresponding observables $F_{\Gamma_L}$ and $F'_{\Gamma'_L}$ commute with $\OL$ accuracy. Similarly, if the distance between $\Gamma_L$ and  $A_L\cap B_L$ is $O(L)$, the observables $\CA_L$ and $\CB_L$ commute with $\OL$ accuracy. We will refer to these and similar properties of $L$-dependent observables $\CA_L$ and $\CB_L$ as "asymptotic localization", where the word "asymptotic" refers to the fact that we study the behavior of the commutators as  $L\ra\infty$.

For a boundary component $\CS$ of a large region $\Gamma$ we define a thickening $\CT\CS$ of $\CS$ as a large region containing all points within a distance of order $L$ from $\CS$ and such that all points of $\CT\CS$ are at a distance of order $L$ from other boundary components.

\begin{figure}
\centering
\begin{tikzpicture}[scale=.5]

\filldraw[color=blue!50, fill=blue!5, dashed, very thick](0,0) circle (4.3);
\filldraw[color=gray, fill=blue!5, very thick](0,0) circle (4);
\filldraw[color=blue!50, fill=white, dashed, very thick](0,0) circle (3.7);
\filldraw[color=blue!50, fill=blue!5, dashed, very thick](1.2,1.2) circle (1.3);
\filldraw[color=gray, fill=blue!5, very thick](1.2,1.2) circle (1);
\filldraw[color=blue!50, fill=white, dashed, very thick](1.2,1.2) circle (0.7);
\filldraw[color=blue!50, fill=blue!5, dashed, very thick](-1.2,-1.2) circle (1.3);
\filldraw[color=gray, fill=blue!5, very thick](-1.2,-1.2) circle (1);
\filldraw[color=blue!50, fill=white, dashed, very thick](-1.2,-1.2) circle (0.7);

\node  at (-1.2,1.2) {$\Gamma$};

\end{tikzpicture}

\caption{A large compact region $\Gamma$ with the boundary components shown as solid gray lines, and a thickening $\CT \p \Gamma$ with the boundaries shown as dashed blue lines.
}
\label{fig:thickening}
\end{figure}
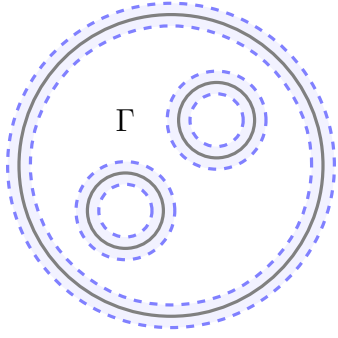

Let $\CS$ be an oriented codimension-one compact surface which is a connected component of the boundary of a large region $\Gamma$. We choose a thickening $\CT\CS$ of $\CS$ and denote $T_{\CT\CS} = T_{(\Gamma\cap\CT\CS) (\bar{\Gamma}\cap\CT\CS)}$. Note that $T_{\CT\CS}=T_{\Gamma\bar\Gamma}+O(L^{-\infty})$.

We claim that the observable $(T_{\CT\CS} + Q_{\Gamma \cap \mathcal{T} \CS})$ has an integer spectrum, up to corrections of order $O(L^{-\infty})$. More precisely, $\exp\left(2\pi i (T_{\CT\CS} + Q_{\Gamma \cap \mathcal{T} \CS})\right)=1+O(L^{-\infty}).$ Indeed, using Lemma \ref{lma:FGA} we can write 
\begin{multline}
T_{\CT \CS} = (\alpha_{F}(Q_{\Gamma})-Q_{\Gamma}) + \OL  = (\alpha_{F_{\CT \p \Gamma}}(Q_{\Gamma})-Q_{\Gamma}) + \OL  = \\ =
(\alpha_{F_{\CT \p \Gamma}}(Q_{\tilde{\Gamma}})-Q_{\tilde{\Gamma}}) + \OL.
\end{multline}
where $\tilde{\Gamma}$ is a compact region such that $\p \tilde{\Gamma}$ contains $\p \Gamma$, and there is a thickening, such that $\CT \p \tilde{\Gamma} \backslash \CT \p \Gamma$ does not intersect $\CT \p \Gamma$. Therefore
\beq
(T_{\CT\CS} + Q_{\Gamma \cap \CT\CS}) + Q_{\tilde{\Gamma} \cap \overline{\CT\CS}} = \alpha_{F_{\CT \p \Gamma}}(Q_{\tilde{\Gamma}}) + \OL.
\eeq
$Q_{\tilde{\Gamma} \cap \overline{\mathcal{T} \CS}}$ commutes with $(T_{\CT\CS} + Q_{\Gamma \cap \mathcal{T} \CS})$ up to terms of order $O(L^{-\infty})$, and both $\alpha(Q_{\tilde{\Gamma}})$ and $Q_{\tilde{\Gamma} \cap \overline{\CT\CS}}$ have integer spectra. This implies the desired result.

Let $\Gamma$ be a large compact region whose boundary $\partial\Gamma$ has a decomposition
$\p \Gamma = \bigcup_{a} \CS_a$. We can choose thickenings $\CT\CS_a$ of all $\CS_a$ such that all $\CT\CS_a$ are far from each other (separated by distances of order $L$). Then we have
\beq\label{QGammaidentity}
\alpha (Q_{\Gamma}) - Q_{\Gamma} = \sum_{a} T_{\CT\CS_a} + \OL.
\eeq
Let us show that
\beq\label{VUidentity}
\alpha( V_{\Gamma}(\phi)) V_{\Gamma}(-\phi) = \prod_{a} Z_{\mathcal{T} \CS_a}(\phi) + \OL,
\eeq
where $Z_{\mathcal{T} \CS_a}(\phi)$ is a unitary almost local observable asymptotically localized on $\CT \CS_a$. First we define $K_{\CT\CS_a}=K_{(\Gamma\cap\CT\CS_a)(\overline{\Gamma}\cap\CT\CS_a)}$. By our assumption on the the thickenings, $K_{\Gamma\overline\Gamma}=\sum_a K_{\CT\CS_a}+O(L^{-\infty})$. Then we get:
\begin{multline}\label{Zdef}
\alpha\left(V_{\Gamma}(\phi)\right) V_{\Gamma}(-\phi) = e^{i \phi (\alpha(Q_{\Gamma}) - \sum_{a}\alpha( K_{\CT\CS_a})} e^{-i \phi (Q_{\Gamma} - \sum_a K_{\CT\CS_a})} +O(L^{-\infty}) = 
\\ = e^{i \phi (Q_{\Gamma} + \sum_{a} T_{\CT\CS_a} - \sum_{a}\alpha( K_{\CT\CS_a}))} e^{-i \phi (Q_{\Gamma} - \sum_a K_{\CT\CS_a})}+O(L^{-\infty}) = \\ = e^{i \phi (\sum_{a} Q_{\Gamma \cap \mathcal{T}\CS_a} + \sum_{a} T_{\CT\CS_a} - \sum_a \alpha( K_{\CT\CS_a})} e^{-i \phi (\sum_{a} Q_{\Gamma \cap \mathcal{T}\CS_a} -\sum_a  K_{\CT\CS_a})}+O(L^{-\infty})= 
\\ = \prod_{a}  e^{i \phi (Q_{\Gamma \cap\mathcal{T}\CS_a} + T_{\CT\CS_a} - \alpha( K_{\CT\CS_a}))} e^{-i \phi (Q_{\Gamma \cap\mathcal{T}\CS_a} - K_{\CT\CS_a})}+O(L^{-\infty}).
\end{multline}
Each factor in the above product is an almost local unitary observable asymptotically localized on some $\CT\CS_a$.

Next we use the following lemma which is a minor variation of Lemma 4.2 from \cite{bachmann2019many}.
\begin{lemma}
\label{lma1v2}
Let $\CU$ be a unitary observable that depends on a parameter $L$, and let $\psi$ be a pure state. Then $|\lal \CU \ral_{\psi}| = 1 - \OL$ is equivalent to  $\lal \CO  \CU \ral_{\psi} - \lal \CO \ral_{\psi} \lal \CU \ral_{\psi} = \OL$ as well as to $\lal  \CU \CO \ral_{\psi} - \lal \CU \ral_{\psi} \lal \CO \ral_{\psi} = \OL$ for all $\CO \in \SA$ with 
$||\CO ||=1$.
\end{lemma}

\begin{proof} 
Let $U$ be an operator representing $\CU$ in the GNS representation for the state $\psi$, and let $P$ be the corresponding vacuum vector projector $P=|0\ral \lal 0 |$.
Then $|\lal \CU \ral_{\psi}| = 1 - \OL$ is equivalent to $||(1-P)U|0\ral|| = \OL$. 
The latter is true if and only if $\lal \CO  \CU \ral_{\psi} - \lal \CO \ral_{\psi} \lal \CU \ral_{\psi} = \OL$ for all $\CO\in\SA$ with $||\CO||=1$. 
Taking complex conjugate, we also obtain equivalence with $\lal  \CU \CO \ral_{\psi} - \lal \CU \ral_{\psi} \lal \CO \ral_{\psi} = \OL$ for all $\CO\in\SA$ with $||\CO||=1$.
\end{proof}
Since by assumption $\alpha$ preserves the ground state and ${\tilde Q}_\Gamma$ does not excite it, we have 
$\lal\alpha(V_{\Gamma}(\phi))V_\Gamma(-\phi) \ral =\lal \alpha(V_{\Gamma}(\phi)) \ral=1$. Then (\ref{Zdef}) and the exponential clustering property for the ground state $\psi$ imply $|\lal Z_{\CT \CS_a} \ral| = 1 - \OL$. Therefore by the above lemma $\lal Z_{\CT \CS_a} \CO\ral=\lal Z_{\CT \CS_a}\ral \lal \CO\ral+\OL$ uniformly in $\CO\in\SA$.

Using this result we obtain a differential equation for $\lal Z_{\CT \CS_a}\ral$:
\begin{multline}
\l -i \frac{d}{d \phi} \r \langle Z_{\mathcal{T} \CS_a}(\phi) \rangle = \\ =
\langle Z_{\mathcal{T} \CS_a}(\phi) e^{i \phi (Q_{\Gamma \cap \mathcal{T} \CS_a}-K_{\CT\CS_a})} \left(T_{\CT\CS_a} + K_{\CT\CS_a} - \alpha(K_{\CT\CS_a}) \right) e^{- i \phi (Q_{\Gamma \cap \mathcal{T} \CS_a}-K_{\CT\CS_a})} \rangle  = \\ = \langle Z_{\mathcal{T} \CS_a}(\phi) \rangle \langle e^{i \phi (Q_{\Gamma \cap \mathcal{T} \CS_a}-K_{\CT\CS_a})} \left(T_{\CT\CS_a} + K_{\CT\CS_a} - \alpha(K_{\CT\CS_a}) \right) e^{- i \phi (Q_{\Gamma \cap \mathcal{T} \CS_a}-K_{\CT\CS_a})} \rangle +O(L^{-\infty})= \\ = \langle Z_{\mathcal{T} \CS_a}(\phi) \rangle \langle e^{i \phi (Q_{\Gamma}-K_{\Gamma \bar{\Gamma}})} \left(T_{\CT\CS_a} + K_{\CT\CS_a} - \alpha(K_{\CT\CS_a}) \right) e^{i \phi (Q_{\Gamma}-K_{\Gamma \bar{\Gamma}})} \rangle +O(L^{-\infty}) = \\ = \langle Z_{\mathcal{T} \CS_a}(\phi) \rangle \langle T_{\CT\CS_a} + K_{\CT\CS_a} - \alpha(K_{\CT\CS_a}) \rangle +O(L^{-\infty}) = \langle Z_{\mathcal{T} \CS_a}(\phi) \rangle \langle T_{\CT\CS_a} \rangle  +O(L^{-\infty}).
\end{multline}
In the fourth line we have used the fact that $(Q_{\Gamma \cap \CT \CS_a} - K_{\CT \CS_a})$ for different $a$ and $Q_{\Gamma \cap \overline{\CT \p \Gamma}}$ commute with each other up to $\OL$ terms. Thus 
\beq
\langle Z_{\mathcal{T} \CS_a} (\phi) \rangle = e^{i \phi \langle T_{\CT\CS_a} \rangle}+O(L^{-\infty}).
\eeq 
Next, consider self-adjoint observables 
$\CO_b=T_{\CT\CS_b} + Q_{\Gamma \cap \mathcal{T} \CS_b}.$ Each of them is asymptotically localized on $\CT\CS_b$ and satisfies $\exp(2\pi i \CO_b)=1+\OL$. Therefore for any boundary component $\CS_a$ one has
\beq\label{trivialid}
e^{2\pi i (\CO_a-\alpha(K_{\CT\CS_a}))}=e^{2\pi i \l \sum_b\CO_b-\alpha(K_{\CT\CS_a})\r}+\OL .
\eeq
On the other hand, eq. (\ref{QGammaidentity}) can be written as
$\sum_b\CO_b+Q_{\Gamma \cap \overline{\CT \p \Gamma}}=\alpha(Q_\Gamma).$ Taking into account that $\exp\l 2\pi i Q_{\Gamma \cap \overline{\CT \p \Gamma}}\r=1$ we get
\beq
e^{2\pi i (\CO_a-\alpha(K_{\CT\CS_a}))}=e^{2\pi i (\alpha(Q_\Gamma)-\alpha(K_{\CT\CS_a}))}+\OL.
\eeq
Therefore
\begin{multline}
\lal Z_{\mathcal{T} \CS_a}(2 \pi) \ral = \langle e^{2 \pi i (Q_{\Gamma \cap\mathcal{T}\CS_a} + T_{\CT\CS_a} - \alpha(K_{\CT\CS_a}))} e^{-2 \pi i (Q_{\Gamma \cap\mathcal{T}\CS_a} - K_{\CT\CS_a})} \rangle  = \\ =  
\langle e^{2 \pi i (\alpha(Q_{\Gamma}) - \alpha(K_{\CT\CS_a}) )} e^{-2 \pi i (Q_{\Gamma \cap\mathcal{T}\CS_a} - K_{\CT\CS_a})} \rangle  +O(L^{-\infty})= \\ = 
\langle \alpha \l e^{2 \pi i (Q_{\Gamma} - K_{\CT\CS_a})} \r e^{-2 \pi i (Q_{\Gamma \cap\mathcal{T}\CS_a} - K_{\CT\CS_a})} \rangle+O(L^{-\infty}) = \\ =
\langle \alpha \l e^{2 \pi i (Q_{\Gamma \cap\mathcal{T}\CS_a} - K_{\CT\CS_a})} \r e^{-2 \pi i (Q_{\Gamma \cap\mathcal{T}\CS_a} - K_{\CT\CS_a})} \rangle+O(L^{-\infty}).
\end{multline}
Now we note that 
\beq
V_{\Gamma}(-2 \pi) = \prod_{a} e^{-2 \pi i (Q_{\Gamma \cap\CT\CS_a} - K_{\CT\CS_a})} + \OL .
\eeq
Since $\langle V_\Gamma(-2\pi)\rangle=1$, the exponential clustering property implies
\beq
|\lal e^{2 \pi i (Q_{\Gamma \cap\CT\CS_a} - K_{\CT\CS_a})} \ral| = 1 - \OL.
\eeq
The above lemma then implies that for any boundary component $\CS_a$ we have  $\langle Z_{\mathcal{T} \CS_a}(2 \pi) \rangle = 1-O(L^{-\infty})$. Therefore $\langle T_{\CT\CS_a} \rangle \in \ZZ$ up to corrections of order $O(L^{-\infty}).$

\begin{remark}\label{rmk:pumpadditive}
For $\alpha = \alpha_1 \circ \alpha_2$, with $\alpha_{1,2}$ generated by $F^{(1,2)}$ and satisfying the properties above, we have
\begin{multline}
\lal T^F_{\CT \CS_a} \ral  = \lal T^{F^{(1)}}_{\CT \CS_a} + \alpha_1(T^{F^{(2)}}_{\CT \CS_a}) \ral + \OL = \\ =
\lal T^{F^{(1)}}_{\CT \CS_a} \ral + \lal T^{F^{(2)}}_{\CT \CS_a} \ral + \OL
\end{multline}
\end{remark}

\begin{remark}\label{rmk:alphapreserves}
The fact that $\alpha$ preserves the ground state was used only to show that $\lal \alpha(V_{\Gamma}(\phi)) \ral = \lal V_{\Gamma}(\phi) \ral$ and 
\beq
\lal \alpha \l e^{2 \pi i (Q_{\Gamma \cap\mathcal{T}\CS_a} - K_{\CT\CS_a})} \r \ral = \lal e^{2 \pi i (Q_{\Gamma \cap\mathcal{T}\CS_a} - K_{\CT\CS_a})} \ral.
\eeq
If these identities are true only up to $\OL$ terms, the approximate integrality of $\lal T_{\CT \CS_a} \ral$ still holds. This fact will be useful in Appendix \ref{app:AvronSeilerSimon}.
\end{remark}

\begin{remark}
This result might seem very general and supply many numerical invariants describing charge transport. In fact for most choices of $\Gamma$ and $\CS$  one finds that $\langle T_{\CT\CS_a} \rangle =\OL$ thanks to the the identity (\ref{2nd property of T}) and very simple topology of $\RR^d$. One exception is the case of one-dimensional systems discussed in the next subsection. Another situation where $\langle T_{\CT\CS_a} \rangle$ has a non-zero limit as $L\ra\infty$ is described in Appendix \ref{app:AvronSeilerSimon}. 

\end{remark}

\subsection{Charge pumping in one dimension}

Let $(\SA, H(s), \psi_s, Q)$ be a differentiable family of one-dimensional gapped lattice systems with a $U(1)$ symmetry for $s \in [0,1]$, such that $H(0)=H(1)$. If the system $(\SA, H(0), \psi_0, Q)$ satisfies the conditions of Ref. \cite{bachmann2012automorphic}, then there is an automorphic equivalence between the states $\psi_s$. In particular, this is the case for systems in the trivial phase. It is expected that all one-dimensional gapped lattice systems with a $U(1)$ symmetry are in the trivial phase \cite{chen2011classification}. Let $\alpha$ be the corresponding quasi-adiabatic automorphism generated by $G(s) = \sum_{j} G_{j}(s)$. By construction it preserves the ground state of $H(0)=H(1)$. Let $\Gamma$ be an interval of length $L$ with the boundary points $(\p \Gamma)_-$ and $(\p \Gamma)_+$. Since $\alpha$ preserves the ground state, we must have $\langle T_{\Gamma\overline{\Gamma}}\rangle=0.$ Choosing thickenings $\CT (\p \Gamma)_-$ and $\CT (\p \Gamma)_+$ which are separated by a distance of order $L$, we see that $\langle T_{\CT (\p \Gamma)_-}\rangle=-\langle T_{\CT (\p \Gamma)_+}\rangle +O(L^{-\infty})$. Taking the limit $L\ra\infty$, we conclude that the quantity $\langle T_{A\overline{A}}\rangle$, where $A=[p,+\infty)$, is independent of the point $p$ and thus is canonically associated to the loop $H(s)$. Furthermore, as shown in the previous section, $\langle T_{A\overline{A}}\rangle$ is an integer. If we vary the loop $(H(s),\psi_s)$ continuously while preserving all the properties above, $\langle T_{A\overline{A}}\rangle$ changes continuously and thus remains constant. Therefore it is a homotopy invariant of the loop $(H(s),\psi_s)$. Its meaning is the charge transported across a point $p$ in the course of one period of quasi-adiabatic evolution.

\section{Quantization of Hall conductance}

\subsection{Hall conductance}

In the remainder of the paper we will study $U(1)$-invariant gapped lattice systems in two dimensions. Given such a system, one can form a current $2 \pi i [\tilde{Q}_j, \tilde{Q}_k]$ which is exact:
\beq\label{QtQt}
2 \pi i [\tilde{Q}_j, \tilde{Q}_k] = - (\p M)_{jk}
\eeq
where
\beq\label{Mrep}
M_{jkl} := \pi i ( [Q_j + \tilde{Q}_j,K_{kl}] + [Q_k + \tilde{Q}_k,K_{lj}] + [Q_l + \tilde{Q}_l ,K_{jk}]).
\eeq
Importantly, one can define the 2-current $M$ for any $U(1)$-invariant state $\psi$ on $\SA$ which has no local spontaneous symmetry breaking. No Hamiltonian needs to be specified.

Consider a point $p$ not in $\Lambda$ and three paths beginning at $p$ and going off to infinity while avoiding $\Lambda$. The paths are assumed to lie in non-overlapping cones with vertex at $p$. These paths divide $\Lambda$ into three noncompact regions which we denote $A,B,C$, see Fig. \ref{fig:magnetization}. They intersect only over the paths, which we denote $AB,BC,CA$. Fixing the orientation of $\RR^2$ also fixes the cyclic order of $A,B,C$.

\begin{figure}
\centering
\begin{tikzpicture}[scale=.3]

\draw[gray, very thick] (0,0) -- (3.4641,2);
\draw[gray, very thick] (0,0) -- (-3.4641,2);
\draw[gray, very thick] (0,0) -- (0,-4);

\node at  (0.6,-0.6) {$p$};
\node  at (0,3) {$C$}; 
\node  at (-1.7321*3/2,-3/2) {$A$};
\node  at (1.7321*3/2,-3/2) {$B$};

\end{tikzpicture}
\caption{Definition of $M_{ABC}$.
}
\label{fig:magnetization}
\end{figure}
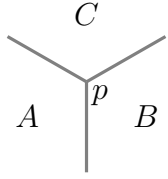

Consider an observable 
\beq
M_{ABC}=\sum_{i\in A,j\in B, k\in C} M_{ijk}.
\eeq 
The infinite sum defining it is norm-convergent, so $M_{ABC}$ is well-defined. We claim that $M_{ABC}$ does not excite the ground state. It is sufficient to show that $\langle M_{ABC} \CO\rangle=0$ for any local observable $\CO$ with a compact localization set and $\langle \CO \rangle=0$. Let $\Gamma_r$ be a disk $B_r(p)$, with the boundary deformed slightly to avoid $\Lambda$. We denote $A'=A\cap\Gamma_r,$ $B'=B\cap\Gamma_r,$ $C'=C\cap\Gamma_r$ and $D=\overline{A'+B'+C'}$, see Fig. 2. Since
\beq
2 \pi i [\tilde{Q}_{A'}, \tilde{Q}_{B'}] = - M_{A'B'C'} - M_{A'B'D},
\eeq
we get 
\beq
\langle M_{A'B'C'} \CO \rangle = - \langle  M_{A'B'D} \CO \rangle = - \langle M_{A'B'D} \rangle \langle \CO \rangle + \Or = \Or,
\eeq
where we used the exponential clustering property. Taking the limit $r\ra\infty$ and noting that $\lim_{r\ra\infty} M_{A'B'C'}=M_{ABC}$, we get that $\langle M_{ABC} \CO\rangle=0$ for any local observable with a compact localization set and $\langle \CO\rangle=0.$ This implies that the same is true for any quasi-local observable with a zero expectation value.

Let 
\beq
h_{jkl} := 2 \lal M_{jkl} \ral
\eeq 
be a 2-current valued in $\RR$. Since $\tilde{Q}_j$ does not excite the ground state, this 2-current is closed, $\p h = 0$. For any three regions $A,B,C$ as above we define
\beq
\sigma := h_{ABC} .
\eeq
The cyclic order of $A,B,C$ is determined by the orientation of $\RR^2$; changing it negates $\sigma$. Since $\p h =0$, the quantity $\sigma$ does not actually depend on the choice of the regions $A,B,C$ or the point $p$. Indeed, if one deforms it by adding some region $D$ to $A$, such that $\p D \cap \p C$ is finite and $[\tilde{Q}_D,\tilde{Q}_C]$ is well-defined, and subtracting it from $B$ (see Fig. \ref{fig:sigma}), one gets
\beq
h_{(A+D)(B-D)C} =  h_{ABC} + h_{DBC} - h_{ADC} = h_{ABC} + (\p h)_{D C} = h_{ABC}.
\eeq
Note that the 2-current $M_{jkl}$ is not uniquely defined by eq. (\ref{QtQt}). One can add any exact 2-current to $M_{jkl}$ or modify $K_{jk}$ in eq. (\ref{Mrep}) to get a new 2-current that satisfies eq. (\ref{QtQt}). However, using the Remark \ref{rmk:Poincare} it is easy to check that $h_{jkl}$ and $\sigma$ are unaffected by these ambiguities, provided the statement of the remark applies to $\psi$. In particular, $\sigma$ can be computed for the ground state of any $U(1)$-invariant gapped lattice 2d system even if the Hamiltonian is not known. 

Let us show that $\sigma$ does not change if one applies to $\psi$ an automorphism $\alpha=\alpha_F(1)$ locally generated by a $U(1)$-invariant 0-chain $F(s)$,  $s\in[0,1]$. As explained in Section 2, the states $\psi(s)=\alpha_F(s)(\psi)$ do not have local spontaneous symmetry breaking. Let $M(s)$ be the 2-chain $M$  computed for the state $\psi(s)$ and $h(s)=\langle M(s)\rangle_{\psi(s)}$. Using (\ref{1st property of T}) and (\ref{eq:Knew}) we get
\begin{multline}
h_{jkl}(1)= \pi i \lal [(Q+\p T^F)_{j}, (K+T^{F})_{kl}] + \text{cyclic permutations of }\{j,k,l\} \ral_\psi
\end{multline}
This equation shows that $h_{ABC}(1)$ does not depend on the behavior of $F$ far from the the point $ABC$. Thus if we replace $F$ with $F_{\Gamma_r}$, where $\Gamma_r$ is a disc of radius $r$ centered at $ABC$, $\sigma$ will only change by an amount of order $\Or$. On the other hand, even after we replace $F$ with $F_{\Gamma_r}$, the new 2-chain $h^r(1)$ still satisfies $\partial h^r(1)=0$. Thus one can compute $\sigma$ using any other point of the plane and three regions meeting at this point. In particular, one can take the point to be far from the disk $\Gamma_r$, so that $ h^r_{ABC}(1)=h_{ABC}(0)+\OL$, $L$ being the distance from the chosen point to $\Gamma_r$. Taking the limit $L\ra\infty$ and $r \to \infty$ we conclude that $h_{ABC}(1)=h_{ABC}(0)$. Thus $\sigma$ is invariant under $U(1)$-invariant locally generated automorphisms.

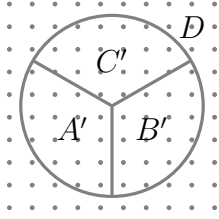
\begin{figure}
\centering
\begin{tikzpicture}[scale=.3]
\draw[gray, very thick] (0,0) -- (3.4641,2);
\draw[gray, very thick] (0,0) -- (-3.4641,2);
\draw[gray, very thick] (0,0) -- (0,-4);

\foreach \x in {-2.25,-1.75,...,2.25}{                           % Two indices running over each
    \foreach \y in {-2.25,-1.75,...,2.25}{                       % node on the grid we have drawn 
    \node[draw,gray,circle,inner sep=.5pt,fill] at (2*\x,2*\y) {}; % Places a dot at those points
    }
}

\draw [gray, very thick] (4,0) arc [radius=4, start angle=0, end angle= 360];
\node  at (0,2) {$C'$}; 
\node  at (-1.7321,-1) {$A'$};
\node  at (1.7321,-1) {$B'$};
\node  at (3.5,3.5) {$D$};

\end{tikzpicture}
\caption{Verifying that the magnetization operator does not excite the ground state.
}
\label{fig:Hall}
\end{figure}

\begin{figure}
\centering
\begin{tikzpicture}[scale=.5]
\draw[gray, very thick] (0,0) -- (3.4641,2);
\draw[gray, very thick] (0,0) -- (-3.4641,2);
\draw[gray, very thick] (0,0) -- (0,-4);

\draw [red, very thick, dashed] (0.5*3.4641, 0.5*2) -- (0.5*3.4641, -4);
\node  at (0,3) {$C$}; 
\node  at (0.86602,-1.5) {$D$};
\node  at (-2.5981,-1.5) {$A$};
\node  at (2.5981,-1.5) {$B$};
\end{tikzpicture}
\caption{Verifying that $\sigma$ does not depend on the choice of the regions $A,B,C$.}
\label{fig:sigma}
\end{figure}
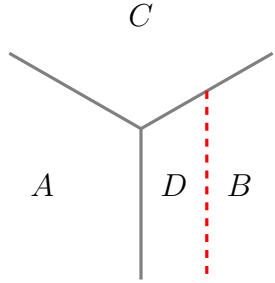

Let us show that $\sigma/2\pi$ is nothing but the zero-temperature Hall conductance. Let $X$ and $Y$ be the right and the upper half-planes, respectively. The Hall conductance is given by the Kubo formula \cite{niu1984quantised}
\beq
\sigma_{Hall}=\sum_{j\in X} \sum_{k\in \bar{X}} \sum_{l \in Y} \sum_{m \in \bar{Y}} i\langle0| J_{jk} (1-P)\frac{1}{H^2}(1-P)J_{lm}|0\rangle -(X\leftrightarrow Y)
\eeq
where $|0\rangle$ is a cyclic vector for the GNS representation, observables are identified with their images in this representation, and $P=|0\rangle \langle 0|$. Although the regions $X,\bar X,Y,\bar Y$ are non-compact, the quadruple sum is absolutely convergent and thus well-defined. To see this, we note that 
for any two observables $\CA$ and $\CB$ we have
\beq
\langle 0| \CA (1-P)\frac{1}{H^2}(1-P)\CB |0\rangle = \langle \SI_\Delta(\CA) \SI_\Delta(\CB) \rangle -  \langle \SI_\Delta(\CA) \rangle \langle \SI_\Delta(\CB) \rangle .
\eeq
Therefore we can rewrite the formula for the Hall conductance in terms of correlators of almost local observables $K_{jk} = \SI_\Delta(J_{jk})$:
\beq
\sigma_{Hall}=\sum_{j\in X} \sum_{k\in \bar{X}} \sum_{l \in Y} \sum_{m \in \bar{Y}} i\langle K_{jk} K_{lm}\rangle -(X\leftrightarrow Y).
\eeq
The sum is absolutely convergent thanks to the exponential decay of correlators in the ground state. In fact, one can re-write this expression as an expectation value of a single almost-local observable:
\beq\label{sigmaHallK}
\sigma_{Hall} = i \langle [K_{X\bar{X}},K_{Y \bar{Y}}] \rangle .
\eeq
While $K_{X\bar X}$ and $K_{Y\bar Y}$ are 0-chains localized on non-compact sets, their commutator is an almost local observable with a well-defined expectation value.

The expression (\ref{sigmaHallK})  does not change if one modifies $X$ by adding or subtracting any compact region $\Gamma$. Indeed, for any compact $\Gamma$ and any $r>0$ one can always find some finite $\Gamma'$ such that the distance between $\Gamma$ and $(Y-\Gamma')$ is of order $r$. From the definition of the current  $K$, we have $K_{(X+\Gamma) (\overline{X+\Gamma})} - K_{X \bar{X}} = Q_{\Gamma} - \tilde{Q}_{\Gamma}$. Therefore the change of the Hall conductance is 
\begin{multline}
i \langle [Q_{\Gamma}-\tilde{Q}_{\Gamma}, K_{Y \bar{Y}}] \rangle = i \langle [Q_{\Gamma}, K_{Y \bar{Y}}] \rangle  = \\ =
i \langle [Q_{\Gamma}, K_{(Y-\Gamma') (\overline{Y-\Gamma'})} + Q_{\Gamma'} - \tilde{Q}_{\Gamma'} ] \rangle = \Or . 
\end{multline}
Here in the last step we used $[Q_\Gamma,Q_{\Gamma'}]=0$ for any two finite regions $\Gamma,\Gamma'$. Taking the limit $r\ra\infty$ we get the desired result.
In the same way one can show that $\sigma_{Hall}$ is not affected  when one modifies $Y$ by a finite region. Modifying $X$ and $Y$ by adding or subtracting infinite regions which lie within non-overlapping cones also does not change $\sigma_{Hall}$ since one can replace them by finite regions of size $r$ up to $\Or$ terms, and then take the limit $r \to \infty$. Therefore instead of $X$ and $Y$ being half-planes, one can take   $X=(C+D)$ and $Y=(A+D)$ with the regions $A,B,C,D$ as shown on Fig. \ref{fig:XY}. 

This configuration has a free parameter $L$ (the distance between two triple points, or equivalently between $B$ and $D$). 
Then we have
\begin{multline}
2 \pi \sigma_{Hall} =  2 \pi i \langle ([K_{CA},K_{AB}] + [K_{AB},K_{BC}] + [K_{BC},K_{CA}]) + \\ + ([K_{AC},K_{CD}] + [K_{CD},K_{DA}] + [K_{DA},K_{AC}]) \rangle +\OL .
\end{multline}
For any three regions $A$,$B$,$C$ as in Fig.~\ref{fig:magnetization} consider a disk $\Gamma_r$ of radius $r$ with the center at the triple point. Let $A'=A\cap \Gamma_r$, etc., as in Fig.~\ref{fig:Hall}. Then
\begin{multline}
\langle [K_{AB}+K_{AC},K_{BC}] \rangle = \langle [K_{A'B'}+K_{A'C'},K_{BC}] \rangle + \Or= \\ = \langle [Q_{A'}-\tilde{Q}_{A'}-K_{A'D},K_{BC}] \rangle + \Or  = \langle [Q_{A'},K_{BC}] \rangle + \Or = \\ = \langle [Q_{A'}+\tilde{Q}_{A'},K_{BC}] \rangle + \Or = \langle [Q_{A}+\tilde{Q}_{A},K_{BC}] \rangle + \Or,
\end{multline}
and since $r$ can be arbitrary, we have
\beq
\langle [K_{AB},K_{BC}] \rangle + \langle [K_{BC},K_{CA}] \rangle = \langle [Q_{A}+\tilde{Q}_{A},K_{BC}] \rangle .
\eeq
Therefore 
\beq
2 \pi i \lal ([K_{CA},K_{AB}] + [K_{AB},K_{BC}] + [K_{BC},K_{CA}]) \ral = \langle M_{ABC}\rangle + \OL,
\eeq
and
\beq
2 \pi \sigma_{Hall} = \langle (M_{ABC} + M_{ACD}) \rangle + \OL = \sigma + \OL
\eeq
Since $L$ can be arbitrary, we have $2 \pi \sigma_{Hall} = \sigma$.

\begin{figure}

\centering
\begin{tikzpicture}[scale=.5]
\draw[gray, very thick] (-1,-1) -- (-1,-3);
\draw[gray, very thick] (-1,-1) -- (-3,-1);
\draw[gray, very thick] (-1,-1) -- (1,1);
\draw[gray, very thick] (1,1) -- (1,3);
\draw[gray, very thick] (1,1) -- (3,1);
\node  at (2,2) {$D$}; 
\node  at (-1,1) {$A$}; 
\node  at (-2,-2) {$B$}; 
\node  at (1,-1) {$C$}; 
\end{tikzpicture}
\caption{A choice for the modified $X$ and $Y$.
}
\label{fig:XY}
\end{figure}
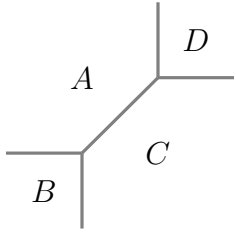

\subsection{Vortices}\label{vortices}

Let us consider three regions $A$,$B$,$C$ meeting at a point $p=ABC$ (see Fig. \ref{fig:vortexA}). As in the previous subsection, we assume that the paths $AB,BC,CA$ lie in non-overlapping open cones with vertex $ABC$.  Let $\upsilon_{ABC}$ be an automorphism of the algebra of observables $\alpha_F(1)$ generated by the 0-chain $F=2\pi (Q_{A}-K_{AB})$.  An equivalent way to define $\upsilon_{ABC}$ is to let $\Gamma_{r}$ be a disk $B_r(p)$, let $A'_r=A \cap \Gamma_{r}$, etc., as in Fig. \ref{fig:magnetization}, and for any observable $\CO$ define 
\begin{equation}
\upsilon_{ABC}(\CO)=\lim_{r\ra\infty} \exp\left(2\pi i (Q_{A'_r}-K_{A'_r B'_r})\right) \CO \exp\left(-2\pi i (Q_{A'_r}-K_{A'_r B'_r})\right).
\end{equation}

The automorphism $\upsilon_{ABC}$ is approximately localized on the path $AB$. This follows from Lemma \ref{lma:FGA} and the fact that $\alpha_{Q_A}(2\pi)$ is the identity automorphism. We will be using this observation many times in what follows.

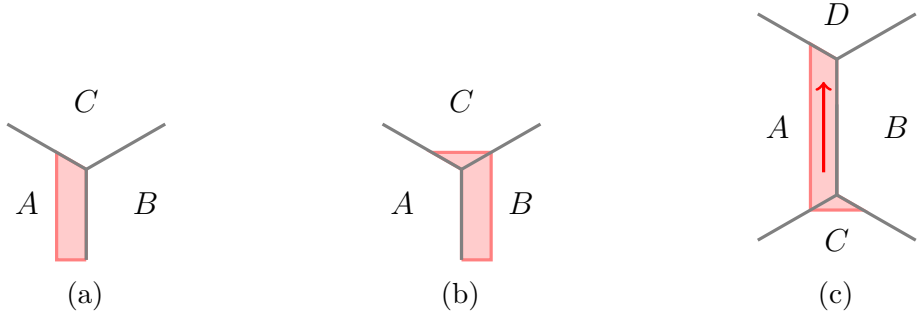
\begin{figure}
    \centering
    \begin{subfigure}[b]{0.3\textwidth}
        \centering
        \begin{tikzpicture}[scale=.3]

         \draw [color=red!50, fill=red!20, very thick] (0,-4) to (0,0) -- (-0.866025*1.5,0.5*1.5) -- (-0.866025*1.5,-4) -- (0,-4);
          \draw[gray, very thick] (0,0) -- (3.4641,2);
          \draw[gray, very thick] (0,0) -- (-3.4641,2);
           \draw[gray, very thick] (0,0) -- (0,-4);
           \node  at (0,3) {$C$}; 
           \node  at (-2.59815,-1.5) {$A$};
           \node  at (2.59815,-1.5) {$B$};
        \end{tikzpicture}
        \caption{}
        \label{fig:vortexA}
    \end{subfigure}
    \begin{subfigure}[b]{0.3\textwidth}
        \centering
        \begin{tikzpicture}[scale=.3]

         \draw [color=red!50, fill=red!20, very thick] (0,-4) to (0,0) -- (0-0.866025*1.5,0.5*1.5) -- (0+0.866025*1.5,0.5*1.5) -- (0+0.866025*1.5,-4) -- (0,-4);
        \draw[gray, very thick] (0,0) -- (0+3.4641,2);
        \draw[gray, very thick] (0,0) -- (0-3.4641,2);
        \draw[gray, very thick] (0,0) -- (0,-4);
        \node  at (0,3) {$C$}; 
        \node  at (0-2.59815,-1.5) {$A$};
        \node  at (0+2.59815,-1.5) {$B$};
        \end{tikzpicture}
        \caption{}
        \label{fig:vortexB}
    \end{subfigure}
    \begin{subfigure}[b]{0.3\textwidth}
        \centering
        \begin{tikzpicture}[scale=.3]

         \draw [color=red!50, fill=red!20, very thick] (0,-3) to (0+3.4641/3,-2/3-3) -- (0-3.4641/3,-2/3-3) -- (0-3.4641/3, 3+0.45*1.5) -- (0,3) -- (0,-3);
        \draw[red, very thick, ->] (0-3.4641/6,-2) -- (0-3.4641/6,2);
        \draw[gray, very thick] (0,-3) -- (0+3.4641,-2-3);
        \draw[gray, very thick] (0,-3) -- (0-3.4641,-2-3);
        \draw[gray, very thick] (0,-3) -- (0,6-3);
        \draw[gray, very thick] (0,-3) -- (0,4-3);
        \draw[gray, very thick] (0,3) -- (0+3.4641,2+3);
        \draw[gray, very thick] (0,3) -- (0-3.4641,2+3);
        \node  at (0,-5) {$C$}; 
        \node  at (0,5) {$D$}; 
        \node  at (0-2.59815,0) {$A$};
        \node  at (0+2.59815,0) {$B$};
        \end{tikzpicture}
        \caption{}
        \label{fig:vortexC}
    \end{subfigure}
    \caption{Creation, annihilation and transport of vortices. The shaded region covers sites, for which operators $Q_j$ are involved.}
    \label{fig:vortex}
\end{figure}

We will say that two states on $\SA$ lie in the same superselection sector if one can be obtained from the other by conjugation with a unitary element of $\SAal$. Note that this differs from both the Doplicher-Haag-Roberts definition and the Buchholz-Fredenhagen definition of superselection sectors as discussed for example in \cite{Haag}. This condition implies that the corresponding GNS representations are unitarily equivalent. Let $\psi_{ABC}$ be a state obtained from the ground state $\psi_0$ by the automorphism $\upsilon_{ABC}$. We claim that the superselection sector of $\psi_{ABC}$ does not depend on the precise location of the paths $AB$, $BC$ and $CA$. More precisely, suppose one chose three non-overlapping open cones with the vertex at $ABC$ which contain the paths $AB$, $BC$ and $CA$. Then changing the paths within these cones will change the state $\psi_{ABC}$ at most by conjugation with an element of $\SAal$. Indeed, we can change the path $BC$ by adding an observable $2 \pi K_{AE}$ to  $2\pi (Q_{A}-K_{AB})$, where $E$ is a (possibly non-compact) region inside the cone of the path $BC$. By Lemma \ref{lma:FA}, the superselection sector is not affected. Changing the path $CA$ corresponds to adding a 0-chain $2 \pi (Q_E-K_{E B})$, where $E$ is a (possibly non-compact) region inside the cone of the path  $CA$. Since $[Q_{E}-K_{EB}, Q_{A}-K_{AB}]$ is an almost local observable, by Lemma \ref{lma:FpX} the superselection sector is not affected. Finally, we can modify the path $AB$ by adding a 0-chain $\tilde{Q}_{E}$, where $E$ is a (possibly non-compact) region inside the cone of the path $AB$. Since $2 \pi (Q_{A}-K_{AB})=2 \pi (\tilde{Q}_A + K_{AC})$ and $[K_{AC},\alpha_{\tilde{Q}_A}(s) (\tilde{Q}_E)]$ is an almost local observable, by Lemma \ref{lma:Ftilde} the superselection sector in unaffected.

The independence of the superselection sector on the choice of the paths has an important consequence. Let $\CA$ be an almost local observable $a$-localized on some site $j$ which is at distance $r$ from the point $ABC$. Let us choose a cone $\Sigma$ with a vertex at $ABC$ and not containing $j$. Given any choice of regions $A,B,C$, we can re-arrange the paths and regions so that the new path $A'B'$ is inside $\Sigma$. Then
\begin{multline}
\lal \CA \ral_{\psi_{ABC}}=\lal \CU \, \CA \, \CU^{-1} \ral_{\psi_{A'B'C'}} + \Or = \\ = \lal \CA \ral_{\psi_{A'B'C'}} + \Or  = \lal \CA \ral_{\psi_0} + \Or.
\end{multline}
Here $\CU \in \SAal$, and we used the localization property of $\upsilon_{A'B'C'}$. 
This implies that almost local observables localized in any cone with vertex $ABC$ and far from $ABC$ cannot detect the presence of a vortex at $ABC$.\ (This statement, however, might not be true if we consider local  observables localized on a ring around $ABC$. In this case one cannot deform the paths such that there is no intersection between the ring and these paths.) In particular, this implies that the state $\psi_{ABC}$ has a finite energy. Such a state can be interpreted as a state with a vortex (unit of magnetic flux) at the point $ABC$.

\begin{remark}
The automorphism $\ups_{ABC}$ has ambiguities related to the choice of the current $K_{jk}$. However, the superselection sector of the state $\psi_{ABC}=\ups_{ABC}(\psi_0)$ is unambiguous. Indeed, modifying $K$ by some $\p N$ leads to an addition of an almost local observable $N_{ABC}$, and by Lemma \ref{lma:FA} does not change the superselection sector. Another way to modify $K$ is to add some 1-current $K'$ that does not excite the ground state. That corresponds to addition of $2 \pi K'_{AB}$ to $2 \pi (Q_A-K_{AB}) = 2 \pi (\tilde{Q}_A+K_{AC})$, and since $[K_{AC},\alpha_{\tilde{Q}_A}(s)(K'_{AB})]$ is an almost local observable, by Lemma \ref{lma:Ftilde} the superselection sector is unchanged. By Remark \ref{rmk:Poincare}, these are the only ambiguities in the definition of $K$.
\end{remark}
Similarly, one can define an automorphism $\bar{\upsilon}_{ABC}$ generated by $2 \pi (Q_{B}+Q_{C}-K_{BA})$ with the same properties (see Fig. \ref{fig:vortexB}) and the state $\bar\psi_{ABC}=\bar\ups_{ABC}(\psi_0)$. Note that $(\bar{\ups}_{ABC} \circ \ups_{ABC})(\psi_{0}) = \psi_{0}$. This follows from $2 \pi (Q_B+Q_C-K_{BA})=2 \pi (Q-(Q_A-K_{AB}))$ and Lemma \ref{lma:FQ}. It is natural to interpret the state produced by the automorphism $\bar{\upsilon}_{ABC}$ as an anti-vortex. By applying automorphisms $\ups$ and $\bar{\ups}$ at different points and choosing the paths so that they do not overlap, one can create several vortices and anti-vortices at different points. The superselection sector of the resulting state does not depend on the choice of the paths, provided the paths are contained in non-overlapping cones.

In general a vortex state cannot be produced by an action of an almost local observable (or even any quasi-local observable) on the ground state. Thus a vortex state may belong to a different superselection sector than the ground state. However, one can create a vortex-anti-vortex pair by acting on the ground state with a unitary almost local observable. For example, suppose one wants to create a vortex at $ABD$ and an anti-vortex at  $BCA$ (see Fig. \ref{fig:vortexC}). This can be accomplished using an automorphism generated by $2 \pi(Q_{A+C} - K_{AB})$. It can be obtained as a limit of automorphisms of the form
\begin{equation}
\CO\mapsto \exp\left(2\pi i (Q_{A'_r+C'_r}-K_{A'_r B'_r})\right) \CO \exp\left(-2\pi i (Q_{A'_r+C'_r}-K_{A'_r B'_r})\right).
\end{equation}
Since $K_{AB}\in\SA$, by Lemma \ref{lma:FA} this automorphism is a conjugation by a unitary observable
\begin{equation}
\alpha_{Q_{A+C}-K_{AB}}(2\pi)\left(e^{-2\pi i K_{AB}}\right).
\end{equation}
In fact, since $K_{AB}\in\SAal$, this observable is almost local. In the following we will be using the notation $e^{2 \pi i (Q_{A+C}-K_{AB})}$ for it.

\begin{figure}
\centering
\begin{tikzpicture}[scale=.5]
\filldraw[color=red!10, fill=red!10, ultra thick] (0,0) -- (2,-4) -- (-2,-4) -- cycle;
\filldraw[color=red!10, fill=red!10, ultra thick] (0,0) -- (3.4641-1,2+3.4641/2) -- (3.4641+1,2-3.4641/2) -- cycle;
\filldraw[color=red!10, fill=red!10, ultra thick] (0,0) -- (-3.4641+1,2+3.4641/2) -- (-3.4641-1,2-3.4641/2) -- cycle;
\draw[red, very thick] (0,0) -- (3.4641+1,2-3.4641/2);
\draw[red, very thick] (0,0) -- (3.4641-1,2+3.4641/2);
\draw[red, very thick] (0,0) -- (-3.4641-1,2-3.4641/2);
\draw[red, very thick] (0,0) -- (-3.4641+1,2+3.4641/2);
\draw[red, very thick] (0,0) -- (-2,-4);
\draw[red, very thick] (0,0) -- (2,-4);

\draw[gray, very thick] (0,0) -- (3.4641,2);
\draw[gray, very thick] (0,0) -- (-3.4641,2);
\draw[gray, very thick] (0,0) -- (0,-4);

\draw[color=blue,ultra thick, ->] (0,2) arc (90:120:2);
\draw[color=blue,ultra thick, ->] (0,2) arc (90:60:2);
\draw[color=blue,ultra thick, ->] (3.4641/2,-2/2) arc (-30:0:2);
\draw[color=blue,ultra thick, ->] (3.4641/2,-2/2) arc (-30:-60:2);
\draw[color=blue,ultra thick, ->] (-3.4641/2,-2/2) arc (210:240:2);
\draw[color=blue,ultra thick, ->] (-3.4641/2,-2/2) arc (210:180:2);

\node  at (0,0.9*3) {$\theta_{C}$}; 
\node  at (-0.9*1.7321*3/2,-0.9*3/2) {$\theta_{A}$};
\node  at (0.9*1.7321*3/2,-0.9*3/2) {$\theta_{B}$};
\node  at (0,1.5*3) {$C$}; 
\node  at (-1.5*1.7321*3/2,-1.5*3/2) {$A$};
\node  at (1.5*1.7321*3/2,-1.5*3/2) {$B$};

\end{tikzpicture}
\caption{Admissible paths $AB$, $BC$ and $CA$ meeting at the point $ABC$.
}
\label{fig:cones}
\end{figure}
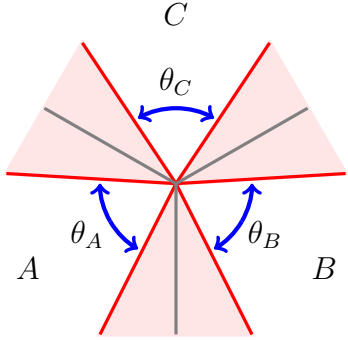

In Section 3 we considered sequences of almost local observables and defined sequences of approximate localization sets for them. In this section we will be dealing with more general infinite sets of almost local observables and it is convenient to generalize the notion of approximate localization sets to them. Let $\{\CA_\alpha,a \in\CI\}$ be an infinite collection of observables labeled by an infinite set $\CI$ and $\{p_\alpha,\alpha\in\CI\}$ be an infinite collection of points of $\Lambda$ labeled by the same set.  We will say that $\{\CA_\alpha\}$ is approximately localized on $\{p_\alpha\}$ if there exists a positive function $f(r)=\Or$ such that for any $r>0$ and any $\alpha\in\CI$ there is $\CA_{\alpha}^{(r)} \in \SA_{B_r(p_{\alpha})}$ such that
$\|\CA_\alpha-\CA_\alpha^{(r)}\|\leq \|\CA_\alpha\| f(r).$ In particular, this condition implies that $\CA_\alpha$ is almost local for all $\alpha$. 

To study the transport of vortices along paths, we will first define the set of vortex configurations and paths of interest. We only consider configurations with a finite number of vortices, so the set of initial and final positions of vortices is always finite. These positions are vertices of a trivalent graph, some of whose edges connect the vertices and some go off to infinity. The paths needed to define vortex states are paths on this graph. For simplicity we will assume that the graph is a tree. The edges of the graph need not be straight lines or segments, but we need to assume that they do not come close to each other. One way to achieve this is the following recursive procedure. Let us fix an angle $\theta_c$ and a vertex $ABC$ (see Fig.\ref{fig:cones}). Removing $ABC$ will cause the graph to fall into three components each of which is itself a tree. We require that each component is contained in a cone with vertex $ABC$ such that the angles between adjacent boundaries of different cones are greater than $\theta_c$ . Then for each component we take the vertex connected to $ABC$ as the basepoint and repeat the procedure. Any trivalent graph which satisfies these requirements will be called admissible.

Let us consider a process (see Fig. \ref{fig:transport}) in which we create a vortex at $(A-E)C(B-C)$ and an anti-vortex at $(A-E)(D-F)E$ and move the vortex to $(B-C)F(D-F)$ along the lines as shown on the figure. We assume that the graph formed by the lines and vertices is admissible. Let $L$ be the smallest distance between the triple points. Let $X_{A}=(Q_{A}-K_{(A-E)(B-C+D)})$ and $X_{B}=(Q_{B}-K_{(B-C)(A+D-F)})$ be 0-chains which generate automorphisms $\alpha_{X_A}(2 \pi)$ and $\alpha_{X_B}(2 \pi)$ corresponding to these movements. Note that $(X_A - Q_{A})$ and $(X_B - Q_{B})$ are almost local observables. Let us denote by $X^r_{A}$ and $X^r_{B}$  the regularized operators $(Q^r_{A'_r}-K_{(A-E)(B-C+D)})$ and $(Q_{B'_r}-K_{(B-C)(A+D-F)})$, correspondingly. Then we have an identity:
\beq\label{eq:composition of transport}
e^{2 \pi i X_B} e^{2 \pi i X_A} |0\ral  
= \l e^{\pi i \lal M_{ABD} \ral} + \OL \r  e^{2 \pi i(X_A+X_B) } |0\ral.
\eeq
where $\OL$ is in an observable with the norm less then some $\OL$ function, which can be chosen the same for any configuration under consideration, $|0\ral$ is the vacuum vector in the ground-state GNS representation, and observables are interpreted as operators using the GNS representation.
 
\begin{proof} 
Let $\CI$ be the set of all labeled configurations of regions and admissible graphs as in Fig. \ref{fig:transport}. To any such a configuration we can attach a self-adjoint almost local observable $2\pi i [X_B,X_A]$. It is easy to see that 
\beq\label{eq:XXM}
2 \pi i [X_{B},X_{A}] = M_{ABD} + \OL ,
\eeq
and thus this collection of observables is approximately localized on the points $ABD$.  The Lieb-Robinson bound implies
\begin{multline}
-i \frac{d}{d \phi} \alpha_{X_{A,B}}(\phi)([X_B,X_A]) = \ad_{X_{A,B}} \l \alpha_{X_{A,B}}(\phi)([X_B,X_A]) \r = \\ =  
\ad_{\tilde{Q}_{A,B}} \l \alpha_{X_{A,B}}(\phi)([X_B,X_A])] \r + \OL ,
\end{multline}
and therefore
\beq
\alpha_{X_{A,B}}(\phi)([X_B,X_A]) = \alpha_{\tilde{Q}_{A,B}}(\phi)([X_B,X_A]) + \OL .
\eeq
This implies that 
\begin{multline}
\l \alpha_{X_B}(\phi)(X_A) - X_A \r =  i \int_{0}^{\phi} ds \:  \alpha_{X_B}(s)([X_B,X_A]) = \\ = 
i \int_{0}^{\phi} ds \:  \alpha_{\tilde{Q}_B}(s)([X_B,X_A]) + \OL
\end{multline}
defines an infinite collection of almost local observables labeled by $\CI$ which is localized at the points $ABD$ and satisfies 
\beq\label{eq:eadXX}
\l \alpha_{X_B}(\phi)(X_A) - X_A \r | 0 \ral = \l \frac{\phi}{2 \pi} \lal M_{ABD} \ral + \OL \r |0\ral.
\eeq
Let $V(\phi)$ and $W(\phi)$ be almost local unitaries
\beq
V(\phi) = \lim_{r \to \infty} e^{- i \phi (Q_{A'_r} + Q_{B'_r})} e^{i \phi (X^r_A + X^r_B)}
\eeq
\beq
W(\phi) = \lim_{r \to \infty} e^{- i \phi (Q_{A'_r} + Q_{B'_r})} e^{i \phi (X^r_B)} e^{i \phi (X^r_A)}
\eeq
which satisfy 
\beq\label{eq:VdV}
V^{\dagger}(\phi) \l -i \frac{d}{d \phi} \r V(\phi) = \alpha^{-1}_{X_A+X_B}(\phi) \l (X_A-Q_{A})+(X_B-Q_{B})\r, 
\eeq
\begin{multline}\label{eq:WdW}
W^{\dagger}(\phi) \l -i \frac{d}{d \phi} \r W(\phi) = \\ = \alpha^{-1}_{X_A}(\phi) \circ \alpha^{-1}_{X_B}(\phi) \l (X_A-Q_{A})+(X_B-Q_{B})+(\alpha_{X_B}(\phi)(X_A)-X_A) \r .
\end{multline}
By comparing eq. (\ref{eq:WdW}) and eq. (\ref{eq:VdV}) and using
\begin{multline}
\l \alpha^{-1}_{X_A}(\phi) \circ \alpha^{-1}_{X_B}(\phi) \l \alpha_{X_B}(\phi)(X_A)-X_A \r \r |0\ral = \\ = \l \alpha^{-1}_{\tilde{Q}_A}(\phi) \circ \alpha^{-1}_{\tilde{Q}_B}(\phi) \l \alpha_{X_B}(\phi)(X_A)-X_A \r  + \OL \r |0\ral = \\ =
\l  \frac{\phi}{2 \pi} \lal M_{ABD} \ral + \OL \r |0\ral ,
\end{multline}
we conclude
\begin{multline}
W^{\dagger}(\phi) \l -i \frac{d}{d \phi} \r W(\phi) |0\ral = \\ =
V^{\dagger}(\phi) \l -i \frac{d}{d \phi} \r V(\phi) |0\ral + \l  \frac{\phi}{2 \pi} \lal M_{ABD}\ral + \OL \r |0\ral.
\end{multline}
Since $V(2 \pi) = e^{2 \pi i (X_A + X_B)}$ and $W(2 \pi ) = e^{2 \pi i X_{B}} e^{2 \pi i X_{A}}$, that implies eq. (\ref{eq:composition of transport}).
\end{proof}

We see that vortex-transport operators for large enough paths without "sharp" turns compose in the expected way except for a phase $e^{\pi i \langle M_{ABD} \rangle}$. Similarly, one can show that if one transports a vortex so that shaded regions intersect (see Fig. \ref{fig:transport2}), one gets a phase $e^{\pi i \langle M_{BAD} \rangle}=e^{-\pi i \langle M_{ABD}\rangle}$. 

\begin{figure}
\centering
\begin{tikzpicture}[scale=0.4]

\draw [color=red!50, fill=red!20, very thick] (0,0) -- (-3.4641,2) -- (-3.4641,3) -- (-3.4641-0.86602,2-0.5) -- (-0.86602,-0.5) -- (-0.86602,-4-0.5) -- (0.86602,-4-0.5) -- (0.86602,-0.5) -- (3.4641+0.86602,2-0.5) -- (3.4641,2) ;
\draw[red, very thick, ->] (-3.4641/2-0.86602/2,2/2-0.5/2) -- (-0.86602/2,-0.5/2) -- (-0.86602/2,-0.5/2-2);
\draw[red, very thick, ->] (0.86602/2,-0.5/2-2) -- (0.86602/2,-0.5/2) -- (3.4641/2+0.86602/2,2/2-0.5/2) ;
\draw[gray, very thick] (0,0) -- (3.4641,2);
\draw[gray, very thick] (3.4641,2) -- (3.4641,4);
\draw[gray, very thick] (3.4641,2) -- (1.5*3.4641, 2-1);
\draw[gray, very thick] (0,0) -- (-3.4641,2);
\draw[gray, very thick] (-3.4641,2) -- (-3.4641,4);
\draw[gray, very thick] (-3.4641,2) -- (-1.5*3.4641,1);
\draw[gray, very thick] (0,0) -- (0,-4);
\draw[gray, very thick] (0,-4) -- (2*0.86602,-4-2*0.5);
\draw[gray, very thick] (0,-4) -- (-2*0.86602,-4-2*0.5);
\node  at (0+0,1.5*2) {$D-F$}; 
\node  at (0-1.5*1.7321,-1.5*1) {$A-E$};
\node  at (0+1.5*1.7321,-1.5*1) {$B-C$};
\node  at (0+0,-6) {$C$};
\node  at (0+1.5*3.4641,3) {$F$};
\node  at (0-1.5*3.4641,3) {$E$};

\draw [color=red!50, fill=red!20, very thick] (15,0) -- (15-3.4641,2) -- (15-3.4641,3) -- (15-3.4641-0.86602,2-0.5) -- (15,-1)  -- (15+3.4641+0.86602,2-0.5) -- (15+3.4641,2) ;
\draw[red, very thick, ->] (15-3.4641/2-0.86602/2,2/2-0.5/2) -- (15,-0.5) -- (15+3.4641/2+0.86602/2,2/2-0.5/2);
\draw[gray, very thick] (15,0) -- (15+3.4641,2);
\draw[gray, very thick] (15+3.4641,2) -- (15+3.4641,4);
\draw[gray, very thick] (15+3.4641,2) -- (15+1.5*3.4641, 2-1);
\draw[gray, very thick] (15,0) -- (15-3.4641,2);
\draw[gray, very thick] (15-3.4641,2) -- (15-3.4641,4);
\draw[gray, very thick] (15-3.4641,2) -- (15-1.5*3.4641,1);
\draw[gray, very thick] (15+0,0) -- (15+0,-4);
\draw[gray, very thick] (15+0,-4) -- (15+2*0.86602,-4-2*0.5);
\draw[gray, very thick] (15+0,-4) -- (15-2*0.86602,-4-2*0.5);
\node  at (15+0,1.5*2) {$D-F$}; 
\node  at (15-1.5*1.7321,-1.5*1) {$A-E$};
\node  at (15+1.5*1.7321,-1.5*1) {$B-C$};
\node  at (15+0,-6) {$C$};
\node  at (15+1.5*3.4641,3) {$F$};
\node  at (15-1.5*3.4641,3) {$E$};

\end{tikzpicture}
\caption{ Transport of vortices.
}
\label{fig:transport}
\end{figure}
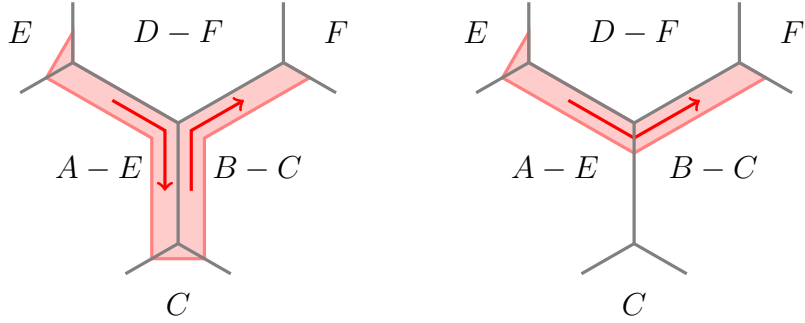

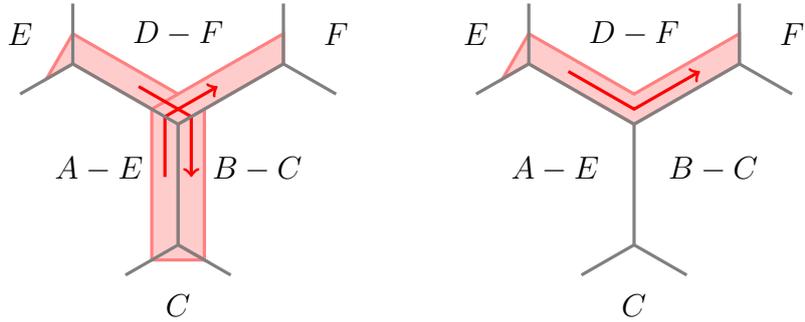
\begin{figure}
\centering
\begin{tikzpicture}[scale=0.4]

\draw [color=red!50, fill=red!20, very thick] (0,0) -- (-3.4641,2) -- (-3.4641-0.86602,2-0.5) -- (-3.4641,3) -- (3.4641/4,2/4)  -- (0.86602,-4-0.5) -- (0,-4) -- (0,0);
\draw [color=red!50, fill=red!20, very thick] (0,0) -- (0,-4) -- (0.86602,-4-0.5) -- (-0.86602,-4-0.5) -- (-3.4641/4,2/4) -- (3.4641,3) -- (3.4641,2) -- (0,0);
\draw[red, very thick, ->] (-3.4641/2+0.86602/2,1+2/2-0.5/2-1/2) -- (+0.86602/2,1-0.5/2-1/2) -- (+0.86602/2,1-0.5/2-2-1/2);
\draw[red, very thick, ->] (-0.86602/2,-0.5/2-2+1/2) -- (-0.86602/2,-0.5/2+1/2) -- (3.4641/2-0.86602/2,2/2-0.5/2+1/2) ;
\draw[gray, very thick] (0,0) -- (3.4641,2);
\draw[gray, very thick] (3.4641,2) -- (3.4641,4);
\draw[gray, very thick] (3.4641,2) -- (1.5*3.4641, 2-1);
\draw[gray, very thick] (0,0) -- (-3.4641,2);
\draw[gray, very thick] (-3.4641,2) -- (-3.4641,4);
\draw[gray, very thick] (-3.4641,2) -- (-1.5*3.4641,1);
\draw[gray, very thick] (0,0) -- (0,-4);
\draw[gray, very thick] (0,-4) -- (2*0.86602,-4-2*0.5);
\draw[gray, very thick] (0,-4) -- (-2*0.86602,-4-2*0.5);
\node  at (0+0,1.5*2) {$D-F$}; 
\node  at (0-1.5*1.7321,-1.5*1) {$A-E$};
\node  at (0+1.5*1.7321,-1.5*1) {$B-C$};
\node  at (0+0,-6) {$C$};
\node  at (0+1.5*3.4641,3) {$F$};
\node  at (0-1.5*3.4641,3) {$E$};

\draw [color=red!50, fill=red!20, very thick] (15,0) -- (15-3.4641,2) -- (15-3.4641-0.86602,2-0.5) -- (15-3.4641,3) -- (15,1) -- (15+3.4641,3) -- (15+3.4641,2) -- (15,0);
\draw[red, very thick, ->] (15-3.4641/2-0.86602/2,2/2-0.5/2+1) -- (15,-0.5+1) -- (15+3.4641/2+0.86602/2,2/2-0.5/2+1);
\draw[gray, very thick] (15,0) -- (15+3.4641,2);
\draw[gray, very thick] (15+3.4641,2) -- (15+3.4641,4);
\draw[gray, very thick] (15+3.4641,2) -- (15+1.5*3.4641, 2-1);
\draw[gray, very thick] (15,0) -- (15-3.4641,2);
\draw[gray, very thick] (15-3.4641,2) -- (15-3.4641,4);
\draw[gray, very thick] (15-3.4641,2) -- (15-1.5*3.4641,1);
\draw[gray, very thick] (15+0,0) -- (15+0,-4);
\draw[gray, very thick] (15+0,-4) -- (15+2*0.86602,-4-2*0.5);
\draw[gray, very thick] (15+0,-4) -- (15-2*0.86602,-4-2*0.5);
\node  at (15+0,1.5*2) {$D-F$}; 
\node  at (15-1.5*1.7321,-1.5*1) {$A-E$};
\node  at (15+1.5*1.7321,-1.5*1) {$B-C$};
\node  at (15+0,-6) {$C$};
\node  at (15+1.5*3.4641,3) {$F$};
\node  at (15-1.5*3.4641,3) {$E$};

\end{tikzpicture}
\caption{Transport of vortices along intersecting paths.
}
\label{fig:transport2}
\end{figure}

\begin{figure}
\centering
\begin{tikzpicture}[scale=.5]
  \node[rotate=90] at (0,0) {
\begin{tikzpicture}[scale=.5]
\draw [color=red!50, fill=red!20, very thick] (0,3) -- (-3.4641,2+3) -- (-3.4641,2+3+1) -- (-1.25*3.4641,+2+3-0.5) -- (0, 3-1) -- (1.25*3.4641,+2+3-0.5) -- (3.4641,2+3);
\draw [color=red!50, fill=red!20, very thick] (0,-3) -- (-3.4641,-2-3) -- (-1.25*3.4641,-2-3+0.5) -- (-3.4641,-2-3-1) -- (0, -4) -- (3.4641,-2-3-1) -- (3.4641,-2-3);
\draw[red, very thick, ->] (-3.4641/2-0.86602/2, 3+2/2-0.5/2) -- (0,3-0.5) -- (3.4641/2+0.86602/2, 3+2/2-0.5/2);
\draw[red, very thick, ->] (-3.4641/2-0.86602/2, -4-2/2+0.5/2) -- (0,-4+0.5) -- (3.4641/2+0.86602/2, -4-2/2+0.5/2);

\draw [color=blue!50, fill=blue!20, very thick, opacity=.5] (0,0) -- (0,3) -- (0+3.4641,2+3) -- (5*0.866025,+2+3-0.5) -- (3.4641,+2+3+1) -- (-0.25*3.4641,3+.5) -- (-0.25*3.4641,-3+.5) -- (-1.25*3.4641,-2-3+0.5) -- (0-3.4641,-2-3) -- (0,-3);
\draw [color=blue!50, fill=blue!20, very thick, opacity=.5] (0,0) -- (0,-3) -- (0+3.4641,-2-3) -- (3.4641,-2-3-1) -- (5*0.866025,-2-3+0.5) --  (0.25*3.4641,-3+.5) -- (0.25*3.4641,3+.5) -- (-3.4641,+2+3+1) --(0-3.4641,2+3) -- (0,3);
\draw[blue, very thick, ->, opacity=.5]  (3.4641/2+0.86602/2+3.4641/8, -4-2/2+0.5/2+1/2+2/8) -- (0+3.4641/8,-4+0.5+1/2+2/8) -- (0+3.4641/8,0) -- (0+3.4641/8,2/8+3) -- (0-3.4641/2,2/2+3+0.5);
\draw[blue, very thick, <-, opacity=.5]  (-3.4641/2-0.86602/2-3.4641/8, -4-2/2+0.5/2+1/2+2/8) -- (0-3.4641/8,-4+0.5+1/2+2/8) -- (0-3.4641/8,0) -- (0-3.4641/8,2/8+3) -- (0+3.4641/2,2/2+3+0.5);

\draw[gray, very thick] (0,-3) -- (0+3.4641,-2-3);
\draw[gray, very thick] (0+3.4641,-2-3) -- (3.4641,-2-3-1);
\draw[gray, very thick] (0+3.4641,-2-3) -- (1.25*3.4641,-2-3+0.5);

\draw[gray, very thick] (0,-3) -- (0-3.4641,-2-3);
\draw[gray, very thick] (-3.4641,-2-3) -- (-3.4641,-2-3-1);
\draw[gray, very thick] (-3.4641,-2-3) -- (-1.25*3.4641,-2-3+0.5);

\draw[gray, very thick] (0,-3) -- (0,6-3);
\draw[gray, very thick] (0,-3) -- (0,4-3);

\draw[gray, very thick] (0,3) -- (0+3.4641,2+3);
\draw[gray, very thick] (3.4641,2+3) -- (3.4641,+2+3+1);
\draw[gray, very thick] (3.4641,2+3) -- (1.25*3.4641,+2+3-0.5);

\draw[gray, very thick] (0,3) -- (0-3.4641,2+3);
\draw[gray, very thick] (-3.4641,2+3) -- (-3.4641,2+3+1);
\draw[gray, very thick] (-3.4641,2+3) -- (-1.25*3.4641,+2+3-0.5);

\end{tikzpicture}
};
\node  at (-5,0) {$C$}; 
\node  at (5,0) {$D$}; 
\node  at (0,2) {$A$};
\node  at (0,-2) {$B$};

\node  at (-5,5) {$2$}; 
\node  at (5,5) {$4$}; 
\node  at (5,-5) {$3$}; 
\node  at (-5,-5) {$1$}; 
\end{tikzpicture}

\caption{One first creates vortex/anti-vortex pairs by operators $12$ and $34$ (shaded in red), and then annihilates them by first applying operator $23$, and then $41$ (shaded in blue).
}
\label{fig:c+cc+c}
\end{figure}
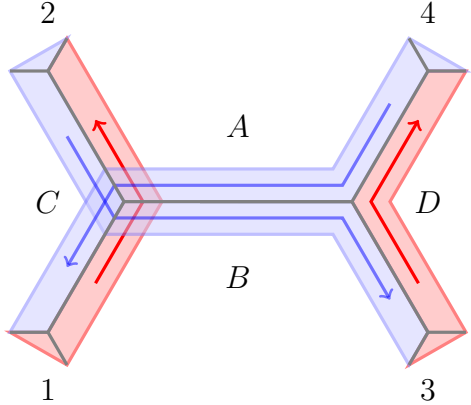

One can create vortices and anti-vortices at the corners of a rectangle by applying transport operators $12$ and $34$ to the vacuum vector $|0\ral$  (see Fig. \ref{fig:c+cc+c}) and then annihilating them in a different order by applying transport operators $23$ followed by $41$. After the application of the operator $23$ one gets the inverse of the application of the operator $41$ times a phase factor $e^{\pi i (\langle M_{ABD} \rangle + \langle M_{ACB} \rangle)}$. Therefore the net result of this operation is multiplication of the vacuum vector by $e^{\pi i \sigma}$ plus $\OL$ corrections.

\subsection{Hall conductance for systems in an invertible phase}\label{HallInv}

In general, one does not expect that one can create a single vortex by applying some almost local or even quasi-local observable to the ground state vector, i.e. the single-vortex state and the ground state can belong to different superselection sectors. However, systems in an invertible phase are special in this regard. 

Let $\left(\SA^{(+)},H^{(+)},\psi^{(+)}_0,Q^{(+)}\right)$ be a $U(1)$-invariant gapped lattice system in an invertible phase. Let $\left(\SA^{(-)},H^{(-)},\psi^{(-)}_0\right)$ be its inverse. Note that we do not require the inverse to have a non-trivial $U(1)$ symmetry. Formally, we may say that it has a $U(1)$ symmetry whose charge is zero. The composite system has a $U(1)$ symmetry whose charge is $Q=Q^{(+)}\otimes 1^{(-)}$. Vortex states for the composite system are defined in the usual manner. We will show that in the composite system vortex states can be obtained from the ground state by conjugation with an almost local unitary. Using that and the relation between the transport properties of vortices and the Hall conductance, we will show that for the original system the Hall conductance is quantized, $\sigma \in \ZZ$. Moreover, we show that for an bosonic system in an invertible phase  $\sigma \in 2 \ZZ$, while for a fermionic system in an invertible phase $\sigma$ is even (odd) if and only if vortices are bosons (fermions).

We start with the bosonic case. Recall that the composite system is defined as follows. Its algebra of observables is  $\SA=\SA^{(+)}\otimes\SA^{(-)}$, its Hamiltonian is  $H=H^{(+)}\otimes 1^{(-)}+1^{(+)}\otimes H^{(-)}$, and its ground state $\Psi$ is a state is defined by $\Psi(\CA^{(+)}\otimes\CA^{(-)})=\psi^{(+)}_0(\CA^{(+)})\psi^{(-)}_0(\CA^{(-)}).$  Let $(\SA,H(s),\Psi(s))$, $s\in [0,1],$ be a path of bosonic  gapped lattice systems connecting the composite system  to a gapped system $(\SA,H(1),\Psi(1)=\tPsi)$ with a factorized ground state $\tPsi$. As discussed in section \ref{QAevolution}, the ground state $\Psi$ for $H(0)$ is automorphically equivalent to $\tPsi$ via an a automorphism $\alpha_G$ locally generated by a 0-chain $G(s) = \sum_{j} G_j (s)$. We denote by $\Pi$ and $\tPi$ the GNS representations of $\SA$ corresponding to the states $\Psi$ and $\tPsi$.

Consider three regions $A$,$B$,$C$ meeting at a point $ABC$ (see Fig. \ref{fig:vortexA}) such that the boundaries $AB,BC,CA$ form an admissible graph. Let $\Ups_{ABC}$ be the vortex inserting automorphism for the composite system $(\SA,H,\Psi,Q)$. It has the form $\Ups_{ABC}=\upsp_{ABC}\otimes 1^{(-)}$.  Let $\tUps_{ABC}$ be the automorphism $\l \alpha_G \circ \Ups_{ABC} \circ \alpha^{-1}_G \r$.

Since the automorphism $\Ups_{ABC}$ is approximately localized on the path $AB$, the same is true about $\tUps_{ABC}$.
Since the superselection sector of the state  $\Ups_{ABC}(\Psi)$ is invariant under re-arranging the paths $AB,BC,CA$, the same is true about the state $\tUps_{ABC}(\tPsi).$ Thus the state $\tUps_{ABC}$ is asymptotically locally indistinguishable from the vacuum state $\tPsi$. By itself, this does not imply $\tUps_{ABC} (\tPsi)$ and $\tPsi$ are in the same superselection sector. For example, topologically non-trivial excitations in a toric code are produced from the ground state precisely in this manner \cite{toriccode}. But since $\tPsi$ is a factorized state, one expects that every pure state which is asymptotically locally indistinguishable from $\tPsi$ is in the same superselection sector. For the states $\tUps_{ABC}(\tPsi)$ and $\tPsi$ this is shown in Appendix \ref{app:trivial} using a  result of T. Matsui \cite{matsui2013boundedness}.

Automorphic equivalence of $\Psi$ and $\tPsi$ (by means of a locally generated automorphism) implies that $\Ups_{ABC}(\Psi)$ is a vector state in the GNS representation for $\Psi$ which can be produced by an almost local unitary $\CV_{ABC}$:
\beq
|\Ups_{ABC} \ral = \Pi(\CV_{ABC}) |0 \ral. %+ \Or
\eeq
Note that because all bounds on the operators used in the construction of vortex states are uniform, there is a function $f(r) = \Or$, such that for any vortex state $\Ups_{ABC}(\Psi)$ constructed using an admissible graph there is an almost local unitary $\CV_{ABC}$ which is $f$-localized.

\begin{figure}
\centering
    \begin{subfigure}[b]{0.3\textwidth}
        \centering
        \begin{tikzpicture}[scale=.3]

         \draw [color=red!50, fill=red!20, very thick] (0,0) -- (-3.4641,2) -- (-3.4641-0.86602,2-0.5) -- (-0.86602,-0.5) -- (-0.86602,-4)  -- (0, - 4) ;
        \draw[red, very thick, <-] (-3.4641/2-0.86602/2,2/2-0.5/2) -- (-0.86602/2,-0.5/2) -- (-0.86602/2,-0.5/2-2);
        \draw[gray, very thick] (0,0) -- (3.4641,2);
        \draw[gray, very thick] (3.4641,2) -- (3.4641,4);
        \draw[gray, very thick] (3.4641,2) -- (1.5*3.4641, 2-1);
        \draw[gray, very thick] (0,0) -- (-3.4641,2);
        \draw[gray, very thick] (-3.4641,2) -- (-3.4641,4);
        \draw[gray, very thick] (-3.4641,2) -- (-1.5*3.4641,1);
        \draw[gray, very thick] (0,0) -- (0,-4);
        \node  at (-1.4*3.4641,1.4*2) {$1$}; 
        \node  at (1.4*3.4641,1.4*2) {$2$}; 
        \end{tikzpicture}
        \caption{}
        \label{fig:exchA}
    \end{subfigure}
    \begin{subfigure}[b]{0.3\textwidth}
        \centering
        \begin{tikzpicture}[scale=.3]
        \draw [color=red!50, fill=red!20, very thick] (0,0) -- (0,-4)  -- (-0.86602,-4) -- (-3.4641/4,2/4) -- (3.4641,3) -- (3.4641,2) -- (0,0);
        \draw[red, very thick, ->] (-0.86602/2,-0.5/2-2+1/2) -- (-0.86602/2,-0.5/2+1/2) -- (3.4641/2-0.86602/2,2/2-0.5/2+1/2) ;
        \draw[gray, very thick] (0,0) -- (3.4641,2);
        \draw[gray, very thick] (3.4641,2) -- (3.4641,4);
        \draw[gray, very thick] (3.4641,2) -- (1.5*3.4641, 2-1);
        \draw[gray, very thick] (0,0) -- (-3.4641,2);
        \draw[gray, very thick] (-3.4641,2) -- (-3.4641,4);
        \draw[gray, very thick] (-3.4641,2) -- (-1.5*3.4641,1);
        \draw[gray, very thick] (0,0) -- (0,-4);
        \node  at (-1.4*3.4641,1.4*2) {$1$}; 
        \node  at (1.4*3.4641,1.4*2) {$2$}; 
        \end{tikzpicture}
        \caption{}
        \label{fig:exchB}
    \end{subfigure}
    \begin{subfigure}[b]{0.3\textwidth}
        \centering
        \begin{tikzpicture}[scale=.3]

        \draw [color=red!50, fill=red!20, very thick] (15,0) -- (15-3.4641,2) -- (15-3.4641-0.86602,2-0.5) -- (15-3.4641,3) -- (15,1) -- (15+3.4641,3) -- (15+3.4641,2) -- (15,0);
        \draw[red, very thick, ->] (15-3.4641/2-0.86602/2,2/2-0.5/2+1) -- (15,-0.5+1) -- (15+3.4641/2+0.86602/2,2/2-0.5/2+1);
        \draw[gray, very thick] (15,0) -- (15+3.4641,2);
        \draw[gray, very thick] (15+3.4641,2) -- (15+3.4641,4);
        \draw[gray, very thick] (15+3.4641,2) -- (15+1.5*3.4641, 2-1);
        \draw[gray, very thick] (15,0) -- (15-3.4641,2);
        \draw[gray, very thick] (15-3.4641,2) -- (15-3.4641,4);
        \draw[gray, very thick] (15-3.4641,2) -- (15-1.5*3.4641,1);
        \draw[gray, very thick] (15+0,0) -- (15+0,-4);
        \node  at (15-1.4*3.4641,1.4*2) {$1$}; 
        \node  at (15+1.4*3.4641,1.4*2) {$2$}; 
        \end{tikzpicture}
        \caption{}
        \label{fig:exchC}
    \end{subfigure}

\caption{The processes corresponding to $\CV_1$, $\CV_2$ and $\CW_{\bar{1} 2}$.
}
\label{fig:exch}
\end{figure}
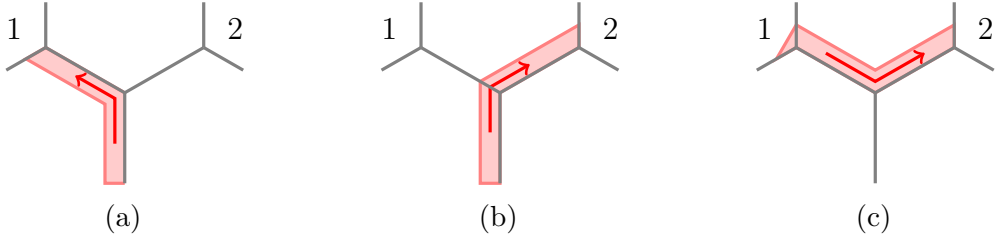

Let us consider an admissible graph depicted on Fig. \ref{fig:exch} with segments connecting triple points having length $L$ . Let $\CV_{1}$ and $\CV_{2}$ be almost local unitary observables creating vortices at points $1$ and $2$ as shown on Fig. \ref{fig:exchA} and Fig. \ref{fig:exchB}, and let $\CW_{\bar{1} 2}$ be a transport operator shown on Fig. \ref{fig:exchC}. Then, using the results from section \ref{vortices}, we have
\begin{multline}
|0 \ral = \Pi \l (\CV_2^{-1} \CV_1)( \CV_2 \CV_1^{-1}) \r |0\ral + \OL  = \\ = e^{-\pi i \sigma/2}
\Pi \l \CV_2^{-1} \CV_1 \CW_{\bar{1} 2} \r |0\ral + \OL =
e^{-\pi i \sigma }
|0\ral  + \OL
\end{multline}
Therefore, for a bosonic spin systems $\sigma \in 2 \ZZ$.

For fermionic systems the arguments are the same, but the unitary equivalence $\CU$ relating the vortex state and the ground state can either preserve or flip fermionic parity. In the former case, the almost local observable $\CV_{ABC}$ has even fermionic parity, and the same arguments as above show that $\sigma\in 2\ZZ$. In the latter case, $\CV_{ABC}$ has odd fermionic parity, and thus operators creating vortices at widely separated points  approximately anti-commute. The above argument then shows that $\sigma$ is an odd integer. Thus vortices are bosons or fermions depending on whether $\sigma$ is even or odd, in agreement with \cite{LevinSenthil}.

\begin{remark} Let us say that a pure state $\psi^{(+)}$ on $\SA^{(+)}$ is in an invertible phase if there is another pure state $\psi^{(-)}$ on $\SA^{(-)}$ and a locally generated automorphism $\alpha_F$ of $\SA = \SA^{(+)} \otimes \SA^{(-)}$ such that the state  $\alpha_F(\psi^{(+)} \otimes \psi^{(-)})$ is factorized. For a factorized state we can always choose a gapped local Hamiltonian $H=\sum_j H_j$, such that it is a ground state of this Hamiltonian. Then $\alpha_F(H)$ is gapped and has $\psi^{(+)} \otimes \psi^{(-)}$ as a ground state. If $\psi^{(+)}$ is invariant under a $U(1)$ symmetry with charge $Q$ on $\SA$, then $\psi^{(+)} \otimes \psi^{(-)}$ is also the ground state of a gapped Hamiltonian $(\alpha_Q(\phi) \circ \alpha_F)(H)$ for any $\phi\in \RR/2\pi\ZZ$.  Therefore $\psi^{(+)}$ is also the ground state of a $U(1)$-invariant gapped Hamiltonian 
\beq
H' = \int_0^{2\pi} (\alpha_Q(\phi) \circ \alpha_F)(H) d \phi .
\eeq
Thus to any $U(1)$-invariant invertible state $\psi^{(+)}$ one can associate a quantized invariant, which is the Hall conductance of the composite system. In the case when $\psi^{(+)}$ satisfies no local spontaneous symmetry breaking condition, this invariant coincides with the Hall conductance of $\psi^{(+)}$. Therefore one does not actually have to use the Hamiltonians $H^{(\pm)}$ of the original and the inverse system anywhere in this section.
\end{remark}

\section{Concluding remarks}

We have shown that both the zero-temperature Hall conductance and the Thouless pump invariant of gapped lattice systems are locally computable. Similar results were recently obtained in \cite{kapustin2019thermal,kapustin2020higherA}. This implies that a 2d gapped system with a nonzero Hall conductance cannot have a gapped interface with a trivial 2d gapped system (or equivalently, cannot have a gapped edge). Similarly, a 1d gapped system with a nonzero Thouless pump invariant cannot have a gapped interface with a trivial 1d system. (In the case of the Thouless pump, constructing the interface involves interpolating both the Hamiltonians and the locally-generated automorphism $\alpha$).

In this paper we adopted a definition of a gapped phase of matter (for a fixed lattice and an algebra of observables) as a homotopy equivalence class of gapped Hamiltonians and their ground states. Another attractive possibility is to keep track of just the states and declare two ground states to be equivalent (and thus in the same gapped phase) if they are related by a locally-generated automorphism of $\SA$. Thanks to the results of \cite{bachmann2012automorphic,moon2020automorphic}, if $(H,\psi)$ and $(H',\psi')$ are equivalent in the former sense, then $\psi$ and $\psi'$ are equivalent in the latter sense. Note also that $\sigma_{Hall}$ is unaffected by locally generated $U(1)$ invariant automorphisms and so can be regarded as an invariant of a gapped phase with $U(1)$ symmetry in this new sense. 

We can completely avoid the usage of the Hamiltonian if we restrict our attention to states in the invertible phase. By the results of Section \ref{HallInv}, if the lattice $\Lambda$ is two-dimensional, the Hall conductance of a pure $U(1)$-invariant state in an invertible phase is well-defined and quantized. Such a Hamiltonian-free definition of an invertible phase could be a useful alternative to the one based on finite-depth local unitary quantum circuits \cite{QImeets} since it guarantees that homotopies in the space of invertible states do not affect the phase to which the system belongs. \\

\noindent
{\bf Acknowledgements:}
This research was supported in part by the U.S.\ Department of Energy, Office of Science, Office of High Energy Physics, under Award Number DE-SC0011632. A.K. was also supported by the Simons Investigator Award. N.S. gratefully acknowledges the support of the Dominic Orr Fellowship at Caltech. \\

\noindent
{\bf Data availability statement:}
Data sharing is not applicable to this article as no new data were created or analyzed in this study.

\appendix
\numberwithin{equation}{section}

\section*{Appendices}

\section{Some technical lemmas}\label{app:lemmas}

The following characterization of almost local observables adapted from  \cite{hastings2010quasi} is sometimes useful. 
\begin{lemma}\label{lma:approximation}
Let $\CA$ be an observable, $j\in\Lambda$ be a site, and $f(r)=\Or$ be a positive monotonically decreasing function.  If for any $k\in\Lambda$ and any $\CB\in\SA_k$ one has $ ||[\CA,\CB]||\leq 2 ||\CA||\cdot ||\CB|| f(\dist(j,k))$, then an observable $\CA\in\SA$ is $h$-localized on a site $j$ for $h(r) = \sup_{j \in \Lambda} \sum_{k \in \bar{B}_r(j)} f(\dist(j,k))=\Or$.
In the opposite direction, if $\CA$ is $f$-localized on a site $j$, then for any $k\in\Lambda$ and any $\CB\in\SA_k$ one has $ ||[\CA,\CB]||\leq 2 ||\CA||\cdot ||\CB|| f(\dist(j,k))$.
\end{lemma}
\begin{proof}
To prove the first statement, note that given any finite $\Gamma\subset\Lambda$ we can define an observable $\mu_\Gamma(\CA)$ by conjugating $\CA$ with $\prod_{k\in\Gamma} U_k$, where $U_k\in\SA_k$ is a unitary, and averaging over all $U_k$. (This operation is also known as partial trace.) The observable $\mu_{\Gamma}(\CA)$ commutes with $\SA_{\Gamma}$ and satisfies $\|\mu_{\Gamma}(\CA)\|\leq \|\CA\|$. Also, it is easy to see that for any local $\CB$ localized outside $\Gamma$ one has $[\mu_\Gamma(\CA),\CB]=\mu_\Gamma([\CA,\CB])$. Let $B_{r,r'}(j)=B_{r'}(j) \backslash B_r(j)$ and let $\CA^{(r,r')}=\mu_{B_{r,r'}(j)}(\CA)$. Using the Cauchy criterion, one can check that $\CA^{(r)}=\lim_{r'\ra\infty} \CA^{(r,r')}$ exists. It is obviously a local operator localized on $B_r(j)$. We have
\begin{multline}
\|\CA-\CA^{(r,r')}\|=\|\int \prod_{k\in B_{r,r'}(j)} dU_k \left({\rm Ad}_{\underset{l\in B_{r,r'}(j)}{\prod} U_l} (\CA)-\CA\right)\|\\
\leq \|\CA\| \sum_{k \in B_{r,r'}(j)} f(\dist(k,j)) \leq \| A \| h(r).
\end{multline}
where $h(r) = \sup_{j \in \Lambda} \sum_{k \in \bar{B}_r(j)} f(\dist(j,k))$. Thus $\|\CA-\CA^{(r)}\| \leq \|\CA\| h(r)$, and therefore $\CA$ is $h$-localized.

To prove the second statement, note that if $\CA$ is $f$-localized on a site $j$, then for any $k \in \Lambda$ with $\dist(j,k)=r$ and any $\CB \in \SA_{k}$ one has $\| [\CA,\CB] \| \leq \| [\CA^{(r)},\CB] \| + 2 \| \CA \| \cdot \| \CB \| f(r) = 2 \| \CA \| \cdot \| \CB \| f(r)$.
\end{proof}

Throughout the paper we use the following version of the Lieb-Robinson bound:

\begin{lemma}\label{lma:LRbound}
Let $F=\sum_j F_j$ be an $f$-local 0-chain bounded by $\|F_j\| \leq C$, and let $\CA$ be an observable  $a$-localized at $p\in\Lambda$. Then $\alpha_F(s)(\CA)$ is $b$-localized at $p$, where the MDP function $b(r)=\Or$ depends on $F$ only through $f(r)$ and $C$ and on $\CA$ only through $a(r)$. 
\end{lemma}

\begin{proof}
Since $\CA$ is $a$-localized one can represent it as a sum $\sum_n \CA_n$ of observable $\CA_n$ local on balls $\Gamma_{r_n}$ of radius $r_n = n$ with the center at the point $p$ such that $\|\CA_n\| \leq \|\CA\|a(r_n)$. Similarly, we have $F_j = \sum_n F_j^{(n)}$ with norms $\|F_j^{(n)}\| \leq C\, f(r_n)$.

For any MDP function $f(r)=\Or$ one can choose constants $c>0$ and $0<\alpha <1$ such that an MDP function $h(r)=c f(r)^{\alpha}/r^{\nu}$ upper-bounds $f(r)$ (cf. \cite{hastings2010quasi}). The function $h(r)$ is reproducing for large enough $\nu$, i.e.
\beq
H_f = \sup_{j,k} \frac{\sum_{l} h(\text{dist}(j,l)) h(\text{dist}(l,k)) }{h(\text{dist}(j,k))} < \infty .
\eeq
More precisely, for a $d$-dimensional lattice the function $h(r)$ is reproducing for $\nu>d$. We also have
\begin{multline}
\|F\|_h := \sup_{k,l \in \Lambda} \: \frac{1}{h(\dist(k,l))} \sum_{j} \sum_{\{n: \, k,l \in \Gamma_{r_n} \}} \|F_j^{(n)}\| \leq \\ \leq
M_f = \sup_{k,l \in \Lambda} \: \frac{C}{h(\dist(k,l))} \sum_{j} \sum_{\{n: \, k,l \in \Gamma_{r_n} \}} f(r_n) < \infty
\end{multline}
where the second sum is over $n$, such that the ball $\Gamma_{r_n}$ with the center at $j$ contains $k$ and $l$. 

Then the automorphism $\alpha_F(s)$ generated by $F$ satisfies the requirement for the Lieb-Robinson bound from \cite{nachtergaele2006propagation}, and therefore for a local observable $\CB \in \text{End}(\CH_k)$ we have
\begin{multline}
\|[\alpha_F(s)(\CA_n),\CB]\| \leq 2 g_f(s) \|\CA_n\| \, \|\CB\| H_f^{-1} \, \sum_{j \in \Gamma_{r_n}} h(\text{dist}(j,k)) \leq \\ \leq 2 \|\CA\|
\|\CB\| \l g_f(s) a(r_n) \, H_f^{-1} \sum_{j \in \Gamma_{r_n}} h(\dist(j,k)) \r
\end{multline}
where $g_f(s) \leq e^{2 H_f M_f s}$. Therefore $\|[\alpha_F(s)(\CA),\CB]\| \leq 2 \|\CA\| \, \|\CB \| \, \tilde{b}(\text{dist}(k,p))$ for a function $\tilde{b}(r)$ that can be chosen to be of order $\Or$. That, together with Lemma \ref{lma:approximation}, implies that $\alpha_F(s)(\CA)$ is $b$-localized for $b(r) = \Or$.
\end{proof}

Let $F = \sum_{j} F_j$ be a 0-chain. For any almost local observable $\CO$ we have \cite{bratteli2012operator2}:
\beq
\alpha_{F}(s)(\CO) = \lim_{r \to \infty} e^{i s F_{\Gamma_r}} \CO e^{- i s F_{\Gamma_r}},
\eeq
where $\Gamma_r$ is a disk with the center at some fixed point. 
% We denote by $\ad_F$ the adjoint action by $F$. For any quasi-local observable $\CO$ we have
% \beq
% \ad_F(\CO) = \sum_j [F_j,\CO]
% \eeq
% which is convergent in the norm for some ordering of sites $j$. 
For a composition of adjoint actions by two different 0-chains $F$ and $X$ we have
\begin{multline}
\ad_X \circ \ad_F (\CO) = \ad_X (\sum_j [F_j,\CO]) = \sum_k \sum_j [X_k,[F_j,\CO]] = \\ = 
\ad_F \circ \ad_X (\CO) + \ad_{\ad_X(F)} (\CO),
\end{multline}
where we have used the fact that we can change the order of summation. Similarly we have
\begin{multline}
\alpha_X(s)\circ \ad_F(\CO) = \alpha_X(s)(\sum_j [F_j,\CO]) = \sum_j [\alpha_X(s)(F_j),\alpha_X(s)(\CO)] = \\ = \ad_{\alpha_X(s)(F)} \circ \alpha_X(s)(\CO).
\end{multline}

\begin{lemma}
\label{lma:FQ}
Let $F = \sum_{j} F_j$ be a 0-chain that is preserved by a  $U(1)$ charge $Q=\sum_{j} Q_j$ (that is, $[Q,F_j]=0$ for all $j$). Then the automorphism $\alpha=\alpha_{2\pi (Q+F)}(1)$ coincides with the automorphism $\alpha'=\alpha_{2\pi F}(1)$.
\end{lemma}

\begin{proof}
We have
\begin{multline}
-i \frac{d}{d t} (\alpha_{Q}(t) \circ \alpha_{F}(t)) = \ad_Q \circ \alpha_{Q}(t) \circ \alpha_{F}(t) +  \alpha_{Q}(t) \circ \ad_F \circ \alpha_{F}(t) = \\ = \ad_Q \circ \alpha_{Q}(t) \circ \alpha_{F}(t) + \ad_{\alpha_Q(t)(F)} \circ \alpha_{Q}(t) \circ \alpha_{F}(t) = \\ =
\ad_{Q+F} \circ \alpha_{Q}(t) \circ \alpha_{F}(t)
\end{multline}
where we have used $\alpha_{Q}(t)(F) = F$. Therefore the automorphism $\alpha_Q(s) \circ \alpha_F(s)$ coincides with $\alpha_{Q+F}(s)$. Since $\alpha_{Q}(2 \pi)$ is the identity automorphism, this implies $\alpha = \alpha'$.
\end{proof}

\begin{lemma}\label{lma:FGA}
Let $F$ be an $f$-local 0-chain and let $G$ be a $g$-local 0-chain approximately localized on a set $\Gamma$. Then for any $j\in\Lambda$  and any observable $\CA\in\SA$ which is $a$-localized on $j\in\Lambda$ one has
$$
||\alpha_{F+G}(s)(\CA)-\alpha_F(s)(\CA)||\leq ||\CA||h(\dist(j,\Gamma)),
$$
where the MDP function $h(r)=\Or$ does not depend on $\Gamma$ and depends on $F,G,\CA$ only through $f,g,a$.
In particular, $\alpha_G$ is approximately localized on $\Gamma$. 
\end{lemma}
\begin{proof}
We compute
\begin{multline}
i\frac{d}{dt} (\alpha_{-(F+G)}(t) \circ \alpha_{F}(t))(\CA)= \alpha_{-(F+G)}(t) \circ \ad_{G} \circ \alpha_{F}(t)(\CA) =\\=
\alpha_{-(F+G)}(t) \circ \alpha_{F}(t) \circ \ad_{\alpha_{-F}(t)(G)}(\CA).
\end{multline}
Integrating this equation from $t=0$ to $t=s$, we get
\beq
\alpha_{-(F+G)}(s) \circ \alpha_{F}(s)(\CA)-\CA=-i\int_0^s dt \l  \alpha_{-(F+G)}(t) \circ \alpha_{F}(t) \circ \ad_{\alpha_{-F}(t)(G)}(\CA) \r.
\eeq
Therefore
\beq
\|\alpha_{-(F+G)}(s)\circ \alpha_{F}(s)(\CA)-\CA\|\leq s\cdot {\rm sup}_{t\in[0,s]} \|\ad_{\alpha_{-F}(t)(G)}(\CA)\|.
\eeq
By Lemma \ref{lma:LRbound} a 0-chain $\alpha_{-F}(t)(G)$ is $b$-local for some $b(r)=\Or$ which depends only on $f,g$ and $t$. This implies the statement of the lemma.
\end{proof}

\begin{lemma}
\label{lma:FA}
Let $A(t)$ be an almost local self-adjoint observable depending on a parameter $t$ and $F = \sum_{j} F_j$ be a 0-chain. Then an automorphism $\alpha_{F+A}(s) \circ \alpha_{-F}(s)$ is a conjugation by an almost local unitary.
\end{lemma}

\begin{proof}

We have
\begin{multline}
-i \frac{d}{dt} \left(\alpha_{F+A}(t) \circ \alpha_{-F}(t) \right) = \alpha_{F+A}(t) \circ \ad_{A(t)} \circ \alpha_{-F}(t) = \\ = \ad_{\alpha_{F+A}(t) (A(t))} \circ \left( \alpha_{F+A}(t) \circ \alpha_{-F}(t) \right).
\end{multline}
Thus $\alpha_{F+A}(s) \circ \alpha_{-F}(s)$ coincides with an automorphism $\alpha_{A'}(s)$, which is generated by $A'(t)=\alpha_{F+A}(t) (A(t))$.
\end{proof}

\begin{lemma}
\label{lma:FpX}
Let $F = \sum_{j} F_j$ and $X = \sum_j X_j$ be 0-chains, such that $i[F,X]$ is an almost local (self-adjoint) observable. Then an automorphism $\alpha_{F}(s) \circ \alpha_{X}(s) \circ \alpha_{-(F+X)}(s)$ is a conjugation by an almost local unitary. 
\end{lemma}
\begin{proof}
Note that an observable
\beq
A(s) = \alpha_{F}(s)(X)-X=i \int_0^s  d t \, \alpha_F(t)([F,X])
\eeq
is almost local and self-adjoint. We have
\begin{multline}
-i \frac{d}{dt} \left( \alpha_{F}(t) \circ \alpha_{X}(t) \right) = \alpha_{F}(t) \circ \ad_{F+X} \circ \alpha_{X}(t) = \\ = \ad_{F+X+A(t)} \circ \left( \alpha_{F}(t) \circ \alpha_{X}(t) \right).
\end{multline}
Thus $\alpha_{F}(s) \circ \alpha_{X}(s) = \alpha_{F+X+A}(s)$, where $\alpha_{F+X+A}$ is generated by $F+X+A(t)$. By Lemma \ref{lma:FA} the automorphism $\alpha_{F+X+A}(s) \circ \alpha_{-(F+X)}(s)$ is a conjugation by an almost local unitary. 
\end{proof}

\begin{lemma}
\label{lma:Ftilde}
Let $F = \sum_{j} F_j$ be a 0-chain and $X = \sum_{j} X_j$ and $Y=\sum_{j} Y_j$ be 0-chains which do not excite the state $\psi$, i.e each $X_j$ and $Y_j$ does not excite $\psi$. Suppose $(\ad_F \circ \alpha_X(s))(Y)$ is an almost local observable. Then the states $\left( \alpha_{F+X}(s) \circ \alpha_{-(F+X+Y)}(s) \right)(\psi)$ and $\psi$ differ by a conjugation by an almost local unitary.
\end{lemma}

\begin{proof}
We have
\begin{multline}
i \frac{d}{dt} \left( \alpha_{F+X}(t) \circ \alpha_{-(F+X+Y)}(t) \right) = \alpha_{F+X}(t) \circ \ad_{Y} \circ \alpha_{-(F+X+Y)}(t) = \\ = \ad_{\alpha_{F+X}(t)(Y)} \circ \alpha_{F+X}(t) \circ \alpha_{-(F+X+Y)}(t).
\end{multline}
Note that
\begin{multline}
A(s) = \alpha_{F+X}(s)(Y)-\alpha_{X}(s)(Y) = \\ = 
\int_{0}^{s} dt  \frac{d}{dt} \left( \alpha_{F+X}(t) \circ \alpha_{X}(s-t)(Y) \right) = \\ =
\int_{0}^{s} dt \left( \alpha_{F+X}(t) \circ \ad_{F} \circ \alpha_{X}(s-t)(Y)  \right)
\end{multline}
is an almost local observable. Since $\alpha_X(t)(Y)$ doesn't excite the state $\psi$, we have
\begin{multline}
\langle \left( \ad_{\alpha_{F+X}(t)(Y)} \circ \alpha_{F+X}(t) \circ \alpha_{-(F+X+Y)}(t) \right) (\CO)  \rangle_{\psi} = \\ = \langle \left( \ad_{A(t)} \circ \alpha_{F+X}(t) \circ \alpha_{-(F+X+Y)}(t) \right) (\CO)  \rangle_{\psi}
\end{multline}
and therefore $\alpha_{F+X}(s) \circ \alpha_{-(F+X+Y)}(s) (\psi)$ is the same as a conjugation of $\psi$ by an almost local unitary, that corresponds to an automorphism $\alpha_{A}(s)$.
\end{proof}

\section{The analogue of the Laughlin argument}\label{app:AvronSeilerSimon}

In this appendix we give an alternative proof of the fact that $\sigma \in \ZZ$ for $U(1)$-invariant invertible gapped lattice systems using a version of the Laughlin argument. This proof is in the spirit of \cite{avron1994}, where it was shown that for gapped systems of free fermions the threading of a unit of flux through a point leads to an inflow of an integer charge.

As was shown in section \ref{HallInv}, a vortex state $\Ups_{ABC}(\Psi)$ for the composite system $(\SA,H,\Psi)$ can be obtained by a conjugation with an almost local unitary observable $\CB_{ABC}$ localized at the point $ABC$. The automorphism $\Ups_{ABC}$ is locally generated by $U(1)$- invariant 0-chain. Although it does not preserve the state $\Psi$, for any MDP function $f(r)=\Or$ and any observable $\CO$ with $\|\CO\|=1$ which is $f$-localized at a point outside of the disk $D$ of radius $L$ with the center at $ABC$ one has $\lal \Ups_{ABC}(\CO) \ral_{\Psi} = \lal \CO \ral_{\Psi} + \OL$. By the Remark \ref{rmk:alphapreserves}, the results of section \ref{ssec:chargepumping1} are applicable for any region $\Gamma$ outside of the disk $D$.

Let $\Gamma$ be an annulus whose inner boundary is a circle $\CS$ of radius $L$ with the center at $ABC$. Then we have
\beq
\exp \l 2 \pi i \lal T_{\CT \CS} \ral_{\Psi} \r = 1 + \OL  .
\eeq
On the other hand,
\begin{multline}
\lal T_{\CT \CS} \ral_{\Psi} = \\ =
\int_{0}^{2 \pi} \lal \alpha_{Q_A-K_{AB}}(\phi)\l  i[K_{(A \cap \bar{\Gamma})B} , Q_{\Gamma}] - i [K_{(A \cap \Gamma)B} , Q_{\bar{\Gamma}}]  \r \ral_{\Psi} + \OL = \\ =
\int_{0}^{2 \pi} \lal \alpha_{\tilde{Q}_A}(\phi)\l  i[K_{(A \cap \bar{\Gamma})B} , Q_{\Gamma}] - i [K_{(A \cap \Gamma)B} , Q_{\bar{\Gamma}}]  \r \ral_{\Psi} + \OL = \\ =
2 \pi i \lal  [K_{(A \cap \bar{\Gamma})B} , Q_{\Gamma}] - [K_{(A \cap \Gamma)B} , Q_{\bar{\Gamma}}] \ral_{\Psi} + \OL = \\ =
2 \pi i \lal [K_{(A \cap \bar{\Gamma})B} + K_{(A \cap \Gamma)B},K_{\Gamma \bar{\Gamma}}] \ral_{\Psi} + \OL = \\ =
2 \pi i \lal [K_{AB},K_{\Gamma \bar{\Gamma}}] \ral_{\Psi} + \OL = - \sigma + \OL .
\end{multline}
In the second line we have used Lemma \ref{lma:FGA}. Thus, taking the limit $L \to \infty$, we get $\sigma \in \mathbb Z$. This computation also shows that $\sigma=2\pi\sigma_{Hall}$ can be interpreted as the net charge transported through a large circle as one inserts a vortex at its origin.

\section{No non-trivial superselection sectors for a factorized state}\label{app:trivial}
In this appendix we show that pure states created by acting on a factorized pure state $\psi_0$ with certain locally generated automorphisms are in the same superselection sector as $\psi_0$. More precisely, they can be obtained from $\psi_0$ by  conjugation with a unitary element of $\SAal$.

\subsection{1d systems}

Let $\Lambda=\ZZ\subset\RR$ and $\SA$ be a quasi-local algebra for the lattice $\Lambda$. We will assume that $d_j^2={\rm dim}\, \SA_j$ grows at most polynomially for $j\ra\infty$. 
Let $\alpha$ be an automorphism of $\SA$ locally generated by a 0-chain $F$ that preserves a factorized state $\psi_0$. Let $A$ be a region $j<0$ and let us set $B=\bar{A}$. We denote by $AB$ the point $j=0$.

Let us first show that the state $\psi = \alpha_{F_{B}}(\psi_0)$ is unitary equivalent to  $\psi_0$. Let $\psi^{(A)}$ and $\psi^{(B)}$ be restrictions of $\psi$ to $\SA_A$ and $\SA_B$. Similarly, let $\psi^{(A)}_0$ and $\psi^{(B)}_0$ be restrictions of $\psi_0$. The states $\psi^{(A)}_0$ and $\psi^{(B)}_0$ are factorized and pure, while the states $\psi^{(A)}$ and $\psi^{(B)}$, in general, are neither. Let $\Gamma_r$ be a disk of radius $r$ with the center at the point $AB$. Lemma \ref{lma:FGA} implies that for any observable $\CA$ which is localized on $A \cap \Gamma_r$ we have $|\lal \CA \ral_{\psi} - \lal \CA \ral_{\psi_0}|=|\lal \alpha_{F_B}(\CA) \ral_{\psi_0}-\lal \CA \ral_{\psi_0}| \leq h(r) \|\CA\|$, where $h(r)=\Or$ is an MDP function which depends only on $F$. Therefore, for any $\epsilon>0$ there is $r$ such that for any observable $\CA$ localized on $A \cap \bar{\Gamma}_r$ we have
\beq\label{BRcorr}
|\lal \CA \ral_{\psi^{(A)}} - \lal \CA \ral_{\psi_0^{(A)}}| \leq \epsilon \|\CA\| .
\eeq
One can express this by saying that the states $\psi^{(A)}$ and $\psi^{(A)}_0$ agree at infinity. By Cor. 2.6.11 of \cite{bratteli2012operator}, the state $\psi^{(A)}$ is quasi-equivalent to $\psi^{(A)}_0$.

A similar argument shows that $\psi^{(B)}$ is quasi-equivalent to $\psi^{(B)}_0$. Indeed, for any observable $\CB$ which is localized on $B\cap\Gamma_r$ we have $|\langle \CB\rangle_\psi-\lal \CB \ral_{\psi_0}|\leq |\langle \alpha_{F_A}^{-1}(\CB)\rangle_{\psi_0}-\lal \CB \ral_{\psi_0}|+ h_1(r) \|\CB\|\leq h_2(r)  \|\CB\|$, where $h_1(r)=\Or$ and $h_2(r)=\Or$ are some MDP functions. Here we used the invariance of $\psi_0$ under $\alpha=\alpha_F$, as well as the fact that   $\alpha_F$ differs from  $\alpha_{F_A}\circ\alpha_{F_B}$ by a conjugation with an almost local observable (see Lemma \ref{lma:FpX}). Then we again use Cor. 2.6.11 of \cite{bratteli2012operator}.

We have shown that the state $\psi^{(A)} \otimes \psi^{(B)}$ on $\SA=\SA_A\otimes\SA_B$ is quasi-equivalent to $\psi^{(A)}_0 \otimes \psi^{(B)}_0 = \psi_0$. In order to prove unitary equivalence of pure states $\psi$ and $\psi_0$ it remains to show that the states $\psi$ and $\psi^{(A)} \otimes \psi^{(B)}$ are quasi-equivalent. For any site $j$ at a distance $r$ from $AB$ and for any observable $\CA \in \SA_j$ we have $|\lal \CA \ral_{\psi} - \lal \CA \ral_{\psi_0}| \leq \|\CA\| f(r)$ where $f(r)=\Or$ is an MDP function which does not depend on $\CA$. Therefore for on-site density matrices $\rho_j$ and $\rho_{0j}$ for the states $\psi$ and $\psi_0$, correspondingly, we have $\| \rho_j - \rho_{0j} \|_{tr.} \leq f(|j|)$, where $\| \cdot \|_{tr.}$ is a trace norm, which is dual to the spectral norm. If $\eps_j = f(|j|)$ is less than $(1/e)$, Fannes inequality implies that the entanglement entropy $S_j$ of the site $j$ satisfies $S_j \leq \eps_j \log(d_j/\eps_j)$. Since $d_j$ grows at most polynomially with $|j|$, we get
\beq
S = \sum_{j \in \Lambda} S_j < \infty .
\eeq
The entanglement entropy of any region $C\subset\Lambda$ is bounded by $S$. Together with Theorem 1.5 from \cite{matsui2013boundedness}, this implies that the states $\psi$ and $\psi^{(A)} \otimes \psi^{(B)}$ are quasi-equivalent. 

It follows that $\psi$ is a vector state $\lal \ups |...| \ups \ral$ in the GNS Hilbert space corresponding to $\psi_0$. The second step is to show that $| \ups \ral$ can be approximated to an $\Or$ accuracy by a vector of the form $\Pi(\CB^{(r)})|0\ral$, where $|0\ral$ is the vacuum vector in the GNS representation of $\psi_0$ and $\CB^{(r)}$ is a local unitary observable with a localization set $\Gamma_r$. 

To see this, note first that for any local $\CA$ localized on some finite region $C \subset \bar{\Gamma}_r$ with $\|\CA\|=1$ we have $\lal \ups| \CA | \ups \ral = \lal 0|\CA| 0 \ral + \Or$. Indeed, for the entropy $S_C$ of the region $C$ we have $S_C \leq g(r)$ for some MDP function $g(r)=\Or$. Let $\rho^{(C)}$ and $\rho^{(C)}_0$ be density matrices of $\psi$ and $\psi_0$ restricted to $C$, correspondingly. Since $(-x \log x) > x$ for $x<1/e$, for $r>r_c$ we have $\| \rho^{(C)} - \tilde{\rho}^{(C)}_0 \|_{tr.} \leq 2 g(r)$ for a restriction $\tilde{\rho}^{(C)}_0$ of some product state $\tilde{\psi}_0$ and some fixed $r_c$ that depends on $g(r)$. On the other hand we have $\|\rho^{(C \cap A)} - \rho^{(C \cap A)}_0 \|_{tr.} \leq h(r)$ and $\|\rho^{(C \cap B)} - \rho^{(C \cap B)}_0 \|_{tr.} \leq h(r)$ for the restrictions of $\rho^{(C)}$ and $\rho^{(C)}_0$ to $A$ and $B$. Therefore $\| \rho^{(C)} - \rho^{(C)}_0 \|_{tr.} = \Or$.

The GNS Hilbert space $\CH$ of $\psi_0$ is a tensor product $\CH_{\Gamma_r} \otimes \CH_{\bar{\Gamma}_r}$ of Hilbert spaces corresponding to the interior and the exterior of $\Gamma_{r}$, and the vacuum vector $|0 \ral$ is factorized as well: $|0 \ral= |0_{\Gamma_r}\ral \otimes |0_{\bar{\Gamma}_r}\ral$. The argument of the preceding paragraph shows that $|\ups \ral = |a_{\Gamma_r} \ral + \Or$ for some $|a_{\Gamma_r} \ral \in \CH_{\Gamma_r} \otimes \text{Span}(|0_{\bar{\Gamma}_{r}})$. This implies that the vector $|\ups \ral$ can be written as follows:
\beq
|\ups \ral = \Pi(\CB^{(r)}) |0 \ral + \Or ,
\eeq
where $\CB^{(r)}$ is a local unitary with a localization set inside $\Gamma_r$. The sequence of vectors $|a_n\ral = |a_{\Gamma_{r_n}} \ral$ for $r_n=n$ converges to $|\ups\ral$. There is a unique special unitary $U_n$ that transforms $|0\ral$ to $|a_{n} \ral$ by performing a rotation in the plane spanned by $|0\ral$ and $|a_{n}\ral$. Clearly, $U_n$ is an image of some local unitary observable localized on $\Gamma_{r_n}$. Since $\lal a_{n+1}|a_{n} \ral = 1 - \Or$, we have $\| U_{n+1}-U_{n} \|=\Or$. Therefore $\lim_{n \to \infty} U_{n}$ converges to an image of an element of some $\CB \in \SAal$.

\subsection{2d systems}

Let $\Lambda\subset\RR^2$ be a lattice and $\SA$ be a quasi-local algebra for $\Lambda$ such that $d_j^2={\rm dim}\, \SA_j$ grows at most polynomially at infinity. Let $\psi_0$ be a factorized pure state on $\SA$. Let us consider three regions $A$,$B$,$C$ meeting at a point $ABC$ (see Fig. 6a). Let $\alpha$ be a locally generated automorphism which is approximately localized on the path $AB$, and such that the superselection sector of a state $\psi = \alpha(\psi_0)$ is invariant under rearrangement of the paths as discussed in section \ref{vortices}. We will first show  that $\psi$ and $\psi_0$ are unitarily equivalent. 
Since both states are pure, it is sufficient to show that they are quasi-equivalent.

Let us fix a cone $\Sigma$ with the center at $ABC$ that contains the region $A$, and let $\bar{\Sigma}$ be its complementary cone. By $\psi^{(\Sigma)}$ and $\psi^{(\bar{\Sigma})}$ we denote the restrictions of the state $\psi$ to $\SA_{\Sigma}$ and $\SA_{\bar{\Sigma}}$, respectively. Let $\Gamma_r$ be a disk of radius $r$ with the center at $ABC$. Similarly to the 1d case, for any $\epsilon>0$ there is an $r>0$ such that for any observable $\CA$ localized on $\bar{\Sigma} \cap \bar{\Gamma}_r$ we have
\beq
|\lal \CA \ral_{\psi^{(\bar{\Sigma})}} - \lal \CA \ral_{\psi_0^{(\bar{\Sigma})}}| \leq \epsilon \|\CA\| .
\eeq 
Invariance of the superselection sector under rearrangements of the paths implies that the same is true for $\psi^{(\Sigma)}$, $\psi^{(\Sigma)}_0$ and observables localized on $\Sigma \cap \bar{\Gamma}_r$. Therefore the states $\psi^{(\Sigma)} \otimes \psi^{(\bar{\Sigma})}$ and $\psi_0$ are quasi-equivalent.

In order to show the quasi-equivalence of states $\psi$ and $\psi^{(\Sigma)} \otimes \psi^{(\bar{\Sigma})}$ we use the same reasoning as in the previous subsection. The same argument shows that the entropy $S_j$ of the restriction of $\psi$ to $j\in\Lambda$ satisfies $S_j \leq g(r)$ for a function $g(r)=\Or$. Here $r$ is the distance between $j$ and the point $ABC$. This implies that the entanglement entropy of any region $\Gamma$ is bounded from above by $S = \sum_{j \in \Lambda} S_j < \infty$. The proof of Theorem 1.5 from \cite{matsui2013boundedness} applies word for word to this setup, with the replacement of left and right half-lines by $\Sigma$ and $\bar{\Sigma}$, respectively, and of intervals $[0,r]$ with $\Sigma \cap \Gamma_r$. Therefore $\psi$ and $\psi^{(\Sigma)} \otimes \psi^{(\bar{\Sigma})}$ are quasi-equivalent.

Finally, having a unitary equivalence between $\psi$ and $\psi_0$, the same reasoning as in the previous subsection shows that $\psi$ is a vector state of the form $\lal 0 |\Pi(\CB^*) ... \Pi(\CB) |0\ral$ for some unitary $\CB \in \SAal$. This is the desired result.

\end{document}